\documentclass[11pt]{article}

\usepackage[a4paper,margin=1in]{geometry}
\usepackage[T1]{fontenc}
\usepackage[utf8]{inputenc}
\usepackage{lmodern}

\usepackage{amsmath,amssymb,amsthm,mathtools}
\usepackage{enumitem}
\usepackage{hyperref}
\usepackage[nameinlink,noabbrev]{cleveref}
\usepackage{xcolor}
\usepackage{microtype}
\usepackage{picture} 
\usepackage{afterpage} 

\usepackage{tikz} 
\usepackage{geometry}
\hypersetup{
	colorlinks=true,
	linkcolor=blue!60!black,
	citecolor=blue!60!black,
	urlcolor=blue!60!black,
	pdftitle={Value Coalition Logic: A Typed Assignment-Based Reconstruction of Coalition Logic},
	pdfauthor={Shanxia Wang}
}

\newtheorem{theorem}{Theorem}[section]
\newtheorem{proposition}[theorem]{Proposition}
\newtheorem{lemma}[theorem]{Lemma}
\newtheorem{corollary}[theorem]{Corollary}

\theoremstyle{definition}
\newtheorem{definition}[theorem]{Definition}
\newtheorem{example}[theorem]{Example}
\newtheorem{remark}[theorem]{Remark}
\newtheorem{convention}[theorem]{Convention}

\newcommand{\tr}{\mathit{tr}}

\newcommand{\VCL}{\mathsf{VCL}}
\newcommand{\CL}{\mathsf{CL}}
\newcommand{\Ag}{\mathsf{Ag}}
\newcommand{\Lang}{\mathcal{L}}
\newcommand{\Coh}{\mathsf{Coh}}

\title{Value Coalition Logic:\\
	A Typed Assignment-Based Reconstruction of Coalition Logic}

\author{Shanxia Wang\thanks{
		\small School of Computer and Information Engineering (School of Artificial Intelligence),\\
		Henan Normal University, Xinxiang, Henan, China\\
		\texttt{wangshanxia@htu.edu.cn}
}}

\date{}

\begin{document}

	\maketitle
	
	\begin{tikzpicture}[remember picture, overlay]
		\node[
		anchor=north west, 
		inner sep=0pt, 
		text width=12cm, 
		font=\footnotesize, 
		black
		] 
		at ([xshift=1.2cm, yshift=-1.8cm]current page.north west) 
		{
			This manuscript has been submitted to the Journal of Logic and Computation (submission ID: JLC 26-088), currently under peer review.
		};
	\end{tikzpicture}

	\begin{abstract}
		We introduce Value Coalition Logic, a typed assignment-based reconstruction of classical coalition logic. The strategic semantics is unchanged: coalitional ability is still interpreted by the standard one-step game-form clause. The change is at the atomic level. Instead of flat propositional valuations, states carry total assignments of values to finitely typed variables. As a result, exhaustivity and mutual exclusion of alternative values are built into the semantics, rather than imposed as external coherence constraints.
		
		We prove that, over each fixed finite typed signature, Value Coalition Logic is truth-equivalent to propositional coalition logic over coherent valuations. This correspondence yields a sound and complete Hilbert-style axiomatisation obtained by adding finite-domain value-coherence axioms to the standard axioms of coalition logic.
		
		The main contribution is structural. Projecting ordinary coalitional ability onto a single value domain yields quotient game forms, projected effectivity families, and strategic value-range hypergraphs. These structures support set-valued strategic exclusion, transversal polarity for disjoint coalitions, exact boundary duality between the empty and grand coalitions, and a measure of residual value indeterminacy. Thus the logic is conservative in its strategic modality, but exposes value-level invariants that are hidden in flat propositional encodings.
	\end{abstract}

	\section{Introduction}
	\label{sec:introduction}
	
	\subsection{Motivation}
	\label{subsec:intro-motivation}
	
	Coalition Logic ($\CL$) is a basic modal framework for reasoning about coalitional power in one-step games~\cite{Pauly02,Pauly02Comp}. In its standard reading, a formula $[C]\varphi$ says that coalition $C$ has a joint action that guarantees $\varphi$, independently of how the agents outside $C$ act. This simple one-step notion of ability has become a semantic core for richer strategic logics, including Alternating-Time Temporal Logic and Strategy Logic~\cite{ATL02,CHP10,MMPV14}.
	
	The usual language of $\CL$ is propositional. States are evaluated against an unstructured set of atomic propositions, and coalitional goals are arbitrary Boolean combinations of such atoms. This is fully adequate at the level of expressive power, but it does not make explicit a common feature of concrete multi-agent specifications: many objectives concern the value of a structured state component. A resource may be assigned to one of several agents, a protocol may be in one of finitely many modes, or a role may be allocated to a particular participant. Such claims are naturally written as value assertions
	\[
	x=c,
	\]
	where $x$ is a state variable and $c$ is an admissible value in its finite domain.
	
	For example, if $x$ denotes the current holder of a resource and
	\[
	D_x=\{1,2,3\},
	\]
	then the assertion $x=2$ says that agent $2$ currently holds the resource. A coalitional claim such as
	\[
	[C](x\in\{1,2\})
	\]
	then says that coalition $C$ can ensure that, after one round, the resource is held by either agent $1$ or agent $2$. This is a set-valued guarantee about a finite state component, not merely a pointwise guarantee of a single propositional atom.
	
	Finite-domain value assertions can of course be encoded propositionally. One may introduce atoms $p_x^c$ for all $c\in D_x$ and read $p_x^c$ as saying that $x$ has value $c$. However, in an unrestricted propositional model the atoms
	\[
	p_x^{c_1},\ldots,p_x^{c_k}
	\]
	are independent. Their intended exactly-one behaviour must be imposed by additional constraints:
	\[
	p_x^{c_1}\lor\cdots\lor p_x^{c_k},
	\qquad
	p_x^{c_i}\to\neg p_x^{c_j}
	\quad(i\neq j).
	\]
	Thus the variable--value organisation of the state description is not primitive in the flat propositional semantics. It can be recovered only after one externally specifies which atoms form the coherent value partition for each variable.
	
	This paper studies a typed assignment-based reconstruction of coalition logic in which this finite value structure is part of the semantic signature. The strategic clause for $[C]$ is left unchanged. What changes is the state-description layer: states are described by total typed assignments rather than by arbitrary propositional valuations. The resulting system, called \emph{Value Coalition Logic} ($\VCL$), is therefore strategically conservative but descriptively structured. Its purpose is not to add new strategic operators or new expressive power over coherent propositional encodings. Its purpose is to make finite value partitions explicit and to analyse the value-level effectivity structures induced by ordinary coalitional ability.
	
	This perspective is close in spirit to component-oriented approaches such as Boolean games and logics of propositional control~\cite{HHMW01,BLLZ06,DHKW08,CLPC05,Gerbrandy06,BH16}. The modelling choice is different, however. $\VCL$ does not assign direct control over propositional letters or variables. It keeps the ordinary coalition modality and replaces only the atomic semantic layer by typed value assertions.
	
	\subsection{Value Coalition Logic}
	\label{subsec:intro-vcl}
	
	A finite typed signature consists of a finite non-empty set of agents $\Ag$, a finite set of state variables $X$, and a finite non-empty domain $D_x$ for each $x\in X$. The atomic formulas of $\VCL$ are sorted value atoms
	\[
	(x{=}c),
	\qquad x\in X,\ c\in D_x.
	\]
	Formulas are obtained from these atoms by Boolean connectives and coalition modalities $[C]$, where $C\subseteq\Ag$.
	
	A $\VCL$ model has the same strategic component as a standard explicit one-step coalition model: a state space, action sets for agents, and a deterministic outcome function. Instead of a propositional valuation, it has a typed assignment map
	\[
	\pi:S\to\prod_{x\in X}D_x.
	\]
	The atomic clause is
	\[
	\mathcal{M},s\models (x{=}c)
	\quad\text{iff}\quad
	\pi(s)(x)=c,
	\]
	while the modal clause for $[C]$ is the usual one-step coalition clause. Hence each state assigns exactly one admissible value to each variable. Exhaustivity and exclusivity of alternative values are semantic consequences of typed assignments, not external propositional side conditions.
	
	For a fixed finite signature, $\VCL$ has a precise correspondence with propositional $\CL$. The translation
	\[
	(x{=}c)\mapsto p_x^c
	\]
	extends homomorphically to all formulas. It yields a truth-preserving correspondence between $\VCL$ models and propositional coalition models whose valuations satisfy the exactly-one constraints for every variable. Conversely, every such coherent propositional model uniquely determines an assignment-based $\VCL$ model over the same strategic frame. This correspondence gives a sound and complete Hilbert-style axiomatisation of $\VCL$ by adding finite-domain value-coherence axioms to the standard axioms of coalition logic.
	
	The correspondence also clarifies the scope of the paper. We do not claim that $\VCL$ has greater expressive power than propositional coalition logic over coherent encodings. Rather, the point is representational and structural: the finite value partitions are built into the semantic signature, and this makes it possible to study the value-level effectivity structures generated by ordinary coalitional ability.
	
	For $A\subseteq D_x$, write
	\[
	(x\in A)
	\quad\text{for}\quad
	\bigvee_{c\in A}(x{=}c).
	\]
	At a model $\mathcal{M}$, state $s$, variable $x$, and coalition $C$, define the projected value-effectivity family
	\[
	\mathcal{N}_x^{\mathcal{M}}(s,C)
	=
	\{\,A\subseteq D_x
	\mid
	\mathcal{M},s\models [C](x\in A)
	\,\}.
	\]
	This family records the value regions to which coalition $C$ can restrict the next value of $x$. It is not a new primitive effectivity semantics. It is obtained by applying the ordinary one-step coalition modality to value-region formulas.
	
	The key construction is the value quotient, or value-projection, of the underlying one-step game form. For each state $s$ and variable $x$, define
	\[
	\mathcal{G}_x^s
	=
	(D_x,\{\Sigma_i\}_{i\in\Ag},o_x^s),
	\qquad
	o_x^s(\gamma)=\pi(o(s,\gamma))(x).
	\]
	Equivalently, this factors the state-valued outcome map through the equivalence relation on successor states
	\[
	t\equiv_x u
	\quad\text{iff}\quad
	\pi(t)(x)=\pi(u)(x).
	\]
	Thus $\mathcal{G}_x^s$ keeps the original agents and actions but identifies successor states that agree on the next value of $x$.
	
	The projected family $\mathcal{N}_x^{\mathcal{M}}(s,C)$ is exactly the ordinary effectivity family of $C$ in this value quotient game form~\cite{Pauly02,GJT13}. Equivalently, each $C$-action determines a possible value range for $x$, and the enforceable value regions are precisely the supersets of such ranges. This yields an operational hypergraph representation of value control.
	
	\subsection{Contributions and Organisation}
	\label{subsec:intro-contributions}
	
	The paper makes five contributions.
	
	\begin{enumerate}[label=\textup{(\arabic*)},leftmargin=*]
		\item It defines $\VCL$ over finite typed signatures. States carry total assignments of values to variables, while the coalition modality keeps the standard one-step semantics of classical coalition logic.
		
		\item It proves a propositional correspondence theorem. Under the translation $(x{=}c)\mapsto p_x^c$, $\VCL$ is truth-equivalent to propositional coalition logic over coherent valuations. This yields a Hilbert-style axiomatisation by adding value-coherence axioms to the usual coalition-logic principles.
		
		\item It introduces value quotient, or value-projection, game forms. For each state $s$ and variable $x$, the original state-valued game form induces a game form over $D_x$, and projected value-effectivity is exactly ordinary game-form effectivity in this quotient.
		
		\item It gives an operational hypergraph representation. The enforceable value regions for a coalition are the upward closure of the value ranges generated by its concrete actions. Minimal enforceable regions are therefore precisely minimal strategic value ranges.
		
		\item It derives structural consequences of value projection: set-valued strategic exclusion, transversal polarity between disjoint coalitions, exact boundary duality between $\varnothing$ and $\Ag$, residual value indeterminacy, and a de re/de dicto contrast for value control.
	\end{enumerate}
	
	These results are deliberately modest in strategic ambition. $\VCL$ does not add temporal operators, quantitative objectives, imperfect information, explicit strategy quantification, or first-order quantifiers. The contribution is to make finite value assignments explicit inside one-step coalition logic and to analyse the value-level structures induced by that typed presentation.
	
	The paper proceeds as follows. Section~\ref{sec:related-work} discusses related work. Section~\ref{sec:vcl} defines the syntax and semantics of $\VCL$. Section~\ref{sec:correspondence-axiomatisation} proves the propositional correspondence and axiomatisation results. Section~\ref{sec:value-quotients} develops value quotients and the hypergraph representation. Section~\ref{sec:structural-results} derives the structural consequences of projected value effectivity. Section~\ref{sec:conclusion} concludes.

	\section{Related Work}
	\label{sec:related-work}
	
	The immediate technical background of this paper is classical Coalition Logic~\cite{Pauly02,Pauly02Comp}. In Pauly's one-step semantics, a formula $[C]\varphi$ expresses that coalition $C$ has a joint action guaranteeing $\varphi$ against all actions of the complementary coalition. This semantics is often formulated in terms of effectivity functions, while explicit game forms provide a concrete strategic representation of such effectivity. We use explicit game forms because their outcome maps can be projected directly to finite value domains. The connection between game forms and effectivity functions, including representation results for playable and truly playable effectivity functions, is well understood~\cite{GJT13}. The present paper remains within this one-step setting. It does not modify the effectivity interpretation of coalitional ability; it modifies the atomic state-description layer on which such ability is evaluated.
	
	A second line of related work concerns structured state descriptions and propositional control. Boolean games and cooperative Boolean games treat propositional variables as strategic components over which agents may have goals, preferences, or control~\cite{HHMW01,BLLZ06,DHKW08}. Logics of propositional control make agency over propositional components explicit either in the object language or in the semantics~\cite{CLPC05,Gerbrandy06,BH16}. $\VCL$ is related to this literature in that it also rejects a completely anonymous view of atomic state descriptions. Its modelling choice, however, is deliberately narrower. The primitive atoms of $\VCL$ are finite-domain value assertions $(x{=}c)$, and the logic does not assign direct control over atoms, propositions, or variables. Coalitional ability remains the ordinary one-step ability of classical coalition logic. Control over values is not postulated; it is induced by the underlying game form and then analysed through value projections.
	
	Richer strategic logics extend coalition logic in temporal, strategic, epistemic, or informational directions. Alternating-Time Temporal Logic adds temporal operators to coalitional ability~\cite{ATL02}, while Strategy Logic makes strategies explicit objects of quantification~\cite{CHP10,MMPV14,MMPV17}. Further variants investigate irrevocable strategies, imperfect information, memory, verification, approximation, and tractable fragments of strategic reasoning~\cite{AGJ07,BJ14,BMMRV17,CLMM18,BLLMY22,KJ24,BFJMM23,BJMM19}. These systems enrich the strategic language, the temporal structure, or the informational assumptions of the model. By contrast, $\VCL$ keeps the one-step coalition modality fixed. Its question is orthogonal: given the ordinary one-step notion of coalitional ability, what structure is obtained when the relevant state components are represented as finite typed variables rather than as an undifferentiated stock of propositional atoms?
	
	There are also related approaches to quantitative, multi-valued, commitment-based, and resource-sensitive strategic reasoning. Multi-valued and quantitative frameworks introduce richer evaluative dimensions into strategic ability~\cite{JKKP20,BG22}. Logics of joint abilities under strategy commitments refine the interaction between what agents can achieve and what they have committed to doing~\cite{LXL20}. Resource-bounded coalition logics and resource-bounded ATL constrain coalitional ability by resources that agents can spend or consume~\cite{Alechina09RBCL,Alechina10RBATL}. These approaches enrich the strategic, evaluative, or resource-sensitive dimension of agency. $\VCL$ instead enriches neither the strategic operator nor the resource model. It makes finite-domain state components primitive and builds their exactly-one behaviour into the semantics.
	
	The closest propositional analogue of $\VCL$ is classical coalition logic over atoms $p_x^c$ together with exactly-one constraints for each finite domain $D_x$. The correspondence developed below shows that, over a fixed finite typed signature, $\VCL$ is truth-equivalent to this coherent propositional encoding. Hence the contribution is not additional expressive power over such encodings. Rather, the contribution is representational and structural: the value partitions are part of the semantic signature, and ordinary coalitional ability can then be projected onto finite value domains. This yields value quotient game forms, projected value-effectivity families, strategic value ranges, transversal constraints, and residual value indeterminacy as explicit objects of analysis.

	\section{Syntax and Semantics of Value Coalition Logic}
	\label{sec:vcl}
	
	This section defines Value Coalition Logic ($\VCL$). The logic keeps the standard one-step coalition modality of classical Coalition Logic, but replaces the flat propositional valuation by a typed assignment of finite values to state variables. Thus the strategic part of the semantics is unchanged; the change is entirely in the atomic state-description layer.
	
	\subsection{One-Step Coalition Frames}
	\label{subsec:vcl-classical}
	
	Let $\Ag$ be a finite non-empty set of agents. For $C\subseteq\Ag$, write
	\[
	\overline C=\Ag\setminus C
	\]
	for the complementary coalition.
	
	A one-step strategic frame consists of a non-empty state space, non-empty action sets for the agents, and a deterministic outcome function. Its strategic component is a tuple
	\[
	(S,\{\Sigma_i\}_{i\in\Ag},o),
	\]
	where $S\neq\varnothing$, each $\Sigma_i\neq\varnothing$, and
	\[
	o:S\times\prod_{i\in\Ag}\Sigma_i\to S
	\]
	maps a current state and a complete action profile to a successor state.
	
	For $C\subseteq\Ag$, put
	\[
	\Sigma_C\coloneqq\prod_{i\in C}\Sigma_i .
	\]
	The empty product $\Sigma_{\varnothing}$ is understood as a singleton, whose unique element is denoted by $\langle\rangle$. If $C,D\subseteq\Ag$ are disjoint and
	\[
	\alpha_C\in\Sigma_C,
	\qquad
	\beta_D\in\Sigma_D,
	\]
	then
	\[
	\alpha_C\sqcup\beta_D\in\Sigma_{C\cup D}
	\]
	denotes the unique joint action extending both partial profiles.
	
	In the usual explicit game-form semantics for Coalition Logic~\cite{Pauly02,Pauly02Comp}, a propositional valuation is added to this strategic component. The modal clause is:
	\[
	\mathcal{M},s\models [C]\varphi
	\]
	iff there exists
	\[
	\alpha_C\in\Sigma_C
	\]
	such that, for every
	\[
	\beta_{\overline C}\in\Sigma_{\overline C},
	\]
	the successor
	\[
	o(s,\alpha_C\sqcup\beta_{\overline C})
	\]
	satisfies $\varphi$. Thus $[C]\varphi$ says that coalition $C$ can guarantee $\varphi$ in one step, independently of the simultaneous choices of the agents outside $C$. This is exactly the strategic clause retained by $\VCL$.
	
	\subsection{Typed Signatures and Value Atoms}
	\label{subsec:vcl-signatures}
	
	\begin{definition}[Typed signature]
		\label{def:typed-signature}
		A \emph{typed signature} is a tuple
		\[
		\Theta=(\Ag,X,\{D_x\}_{x\in X}),
		\]
		where $\Ag$ is a finite non-empty set of agents, $X$ is a finite set of state variables, and each $D_x$ is a finite non-empty domain of admissible values for $x$.
	\end{definition}
	
	Elements of
	\[
	\prod_{x\in X}D_x
	\]
	are called \emph{typed assignments}. If
	\[
	\eta\in\prod_{x\in X}D_x
	\]
	and $x\in X$, then $\eta(x)\in D_x$ is the value assigned to $x$ by $\eta$. When $X=\varnothing$, the empty product is understood as the singleton containing the empty assignment. This degenerate case is harmless; the value-based constructions below are non-trivial only when at least one variable is present.
	
	\begin{definition}[$\VCL$ formulas]
		\label{def:vcl-syntax}
		Let
		\[
		\Theta=(\Ag,X,\{D_x\}_{x\in X})
		\]
		be a typed signature. The language $\Lang_{\VCL}(\Theta)$ is generated by
		\[
		\varphi
		\;::=\;
		\top
		\mid
		(x{=}c)
		\mid
		\neg\varphi
		\mid
		(\varphi\land\psi)
		\mid
		[C]\varphi,
		\]
		where $x\in X$, $c\in D_x$, and $C\subseteq\Ag$.
	\end{definition}
	
	The connectives $\lor$, $\to$, $\leftrightarrow$ and the constant $\bot$ are defined as usual. We write $\Lang_{\VCL}$ when the signature is clear.
	
	\begin{convention}[Sorted value atoms]
		\label{conv:sorted-value-atoms}
		For each $x\in X$ and $c\in D_x$, the expression $(x{=}c)$ is a primitive atom of $\Lang_{\VCL}(\Theta)$. If $c\notin D_x$, then $(x{=}c)$ is not a well-formed formula.
		
		The equality symbol in $(x{=}c)$ is therefore syntactic notation for a sorted value assertion. $\VCL$ has no independent identity predicate, no term grammar, and no quantification over values. If the same symbol $c$ occurs in two domains $D_x$ and $D_y$, then $(x{=}c)$ and $(y{=}c)$ are distinct sorted atoms. Equivalently, the atomic vocabulary may be viewed as a family of indexed atoms $p_x^c$ with intended reading $x=c$.
	\end{convention}
	
	\subsection{$\VCL$ Models}
	\label{subsec:vcl-models}
	
	\begin{definition}[$\VCL$ model]
		\label{def:vcl-model}
		Let
		\[
		\Theta=(\Ag,X,\{D_x\}_{x\in X})
		\]
		be a typed signature. A \emph{$\VCL$ model} over $\Theta$ is a tuple
		\[
		\mathcal{M}=(S,\{\Sigma_i\}_{i\in\Ag},o,\pi),
		\]
		where:
		\begin{itemize}[leftmargin=*]
			\item $S$ is a non-empty set of states;
			\item each $\Sigma_i$ is a non-empty set of actions available to agent $i$;
			\item
			\[
			o:S\times\Sigma_{\Ag}\to S
			\]
			is a deterministic one-step outcome function;
			\item
			\[
			\pi:S\to\prod_{x\in X}D_x
			\]
			is a total typed assignment map.
		\end{itemize}
	\end{definition}
	
	Thus a $\VCL$ model has the same strategic component as an explicit one-step coalition model. The only difference is that the propositional valuation is replaced by the assignment map $\pi$.
	
	The map $\pi$ is not required to be injective or surjective. Distinct states may carry the same typed assignment while differing in their strategic transition behaviour, and some assignments in $\prod_{x\in X}D_x$ need not be realised by any state. This mirrors ordinary propositional models, where distinct states may agree on all atomic facts but differ strategically.
	
	\subsection{Truth Definition}
	\label{subsec:vcl-truth}
	
	\begin{definition}[Truth definition]
		\label{def:vcl-semantics}
		Let
		\[
		\mathcal{M}=(S,\{\Sigma_i\}_{i\in\Ag},o,\pi)
		\]
		be a $\VCL$ model over $\Theta$. The satisfaction relation $\mathcal{M},s\models\varphi$ is defined by induction:
		\begin{align*}
			\mathcal{M},s\models\top
			&\quad\text{always},\\
			\mathcal{M},s\models(x{=}c)
			&\quad\text{iff}\quad \pi(s)(x)=c,\\
			\mathcal{M},s\models\neg\varphi
			&\quad\text{iff}\quad \mathcal{M},s\not\models\varphi,\\
			\mathcal{M},s\models\varphi\land\psi
			&\quad\text{iff}\quad
			\mathcal{M},s\models\varphi
			\text{ and }
			\mathcal{M},s\models\psi,\\
			\mathcal{M},s\models[C]\varphi
			&\quad\text{iff}\quad
			\exists\,\alpha_C\in\Sigma_C
			\text{ such that }
			\forall\,\beta_{\overline C}\in\Sigma_{\overline C},\\
			&\hspace{3.1cm}
			\mathcal{M},o(s,\alpha_C\sqcup\beta_{\overline C})\models\varphi .
		\end{align*}
	\end{definition}
	
	The modal clause is exactly the ordinary one-step coalition clause. Coalition $C$ chooses a joint action for its members, and $\varphi$ must hold at every immediate successor compatible with arbitrary simultaneous actions of the complementary coalition. The distinctive feature of $\VCL$ lies only in the atomic clause: the truth of $(x{=}c)$ is determined by the value assigned to $x$ at the current state.
	
	\subsection{Assignment Coherence}
	\label{subsec:vcl-coherence}
	
	Because $\pi$ assigns a total value profile to every state, the alternative values of each variable are jointly exhaustive and pairwise exclusive. This is the basic semantic fact that later corresponds to propositional exactly-one constraints.
	
	\begin{theorem}[Assignment coherence]
		\label{thm:assignment-coherence}
		Let
		\[
		\mathcal{M}=(S,\{\Sigma_i\}_{i\in\Ag},o,\pi)
		\]
		be a $\VCL$ model over $\Theta$. For every state $s\in S$ and every variable $x\in X$:
		\begin{enumerate}[label=\textup{(\roman*)},leftmargin=*]
			\item \textup{(\emph{Exhaustivity})}
			\[
			\mathcal{M},s\models \bigvee_{c\in D_x}(x{=}c).
			\]
			
			\item \textup{(\emph{Exclusivity})}
			For all distinct $c,d\in D_x$,
			\[
			\mathcal{M},s\models (x{=}c)\to\neg(x{=}d).
			\]
		\end{enumerate}
		Consequently,
		\[
		\mathcal{M},s\models
		\Bigl(\bigvee_{c\in D_x}(x{=}c)\Bigr)
		\land
		\bigwedge_{\substack{c,d\in D_x\\ c\neq d}}
		\neg\bigl((x{=}c)\land(x{=}d)\bigr).
		\]
	\end{theorem}
	
	\begin{proof}
		Fix $s\in S$ and $x\in X$. Since
		\[
		\pi(s)\in\prod_{y\in X}D_y,
		\]
		there is a unique $c^\ast\in D_x$ such that
		\[
		\pi(s)(x)=c^\ast.
		\]
		By the atomic truth clause,
		\[
		\mathcal{M},s\models(x{=}c^\ast),
		\]
		and hence
		\[
		\mathcal{M},s\models\bigvee_{c\in D_x}(x{=}c).
		\]
		
		For exclusivity, let $c,d\in D_x$ with $c\neq d$. If
		\[
		\mathcal{M},s\models(x{=}c),
		\]
		then $\pi(s)(x)=c$. Hence $\pi(s)(x)\neq d$, so
		\[
		\mathcal{M},s\not\models(x{=}d).
		\]
		Therefore
		\[
		\mathcal{M},s\models (x{=}c)\to\neg(x{=}d).
		\]
		The final displayed formula follows by classical propositional reasoning.
	\end{proof}
	
	\begin{remark}[Coherence and propositional encodings]
		\label{rem:coherence-significance}
		Theorem~\ref{thm:assignment-coherence} is not an additional strategic principle. It follows solely from interpreting states as total typed assignments. In a flat propositional presentation, the same behaviour can be simulated by atoms $p_x^c$ together with exactly-one constraints for every variable $x$. In $\VCL$, these constraints are built into the semantic interpretation of the atomic layer.
	\end{remark}

	\section{Propositional Correspondence and Axiomatisation}
	\label{sec:correspondence-axiomatisation}
	
	This section makes precise the sense in which $\VCL$ is a typed reconstruction of propositional Coalition Logic. Since all value domains are finite, every value atom $(x{=}c)$ can be represented by a propositional atom $p_x^c$. This representation is faithful exactly over those propositional models in which, for each variable $x$, the atoms $\{p_x^c\mid c\in D_x\}$ form an exactly-one value partition at every state.
	
	Throughout this section, fix a finite typed signature
	\[
	\Theta=(\Ag,X,\{D_x\}_{x\in X}).
	\]
	
	\subsection{Coherent Propositional Encodings}
	\label{subsec:encoding}
	
	For each $x\in X$ and $c\in D_x$, let $p_x^c$ be a distinct propositional atom, and put
	\[
	P_\Theta
	\coloneqq
	\{\,p_x^c \mid x\in X,\ c\in D_x\,\}.
	\]
	Let $\Lang_{\CL}(P_\Theta)$ be the ordinary propositional coalition language over $P_\Theta$.
	
	\begin{definition}[Propositional translation]
		\label{def:translation}
		The translation
		\[
		\tr:\Lang_{\VCL}(\Theta)\to\Lang_{\CL}(P_\Theta)
		\]
		is defined recursively by
		\begin{align*}
			\tr(\top) &\coloneqq \top,\\
			\tr(x{=}c) &\coloneqq p_x^c,\\
			\tr(\neg\varphi) &\coloneqq \neg\tr(\varphi),\\
			\tr(\varphi\land\psi) &\coloneqq \tr(\varphi)\land\tr(\psi),\\
			\tr([C]\varphi) &\coloneqq [C]\tr(\varphi).
		\end{align*}
	\end{definition}
	
	Thus $\tr$ is homomorphic for Boolean connectives and coalition modalities. Its only non-trivial action is to replace each sorted value atom by the corresponding propositional atom.
	
	\begin{definition}[Coherence constraints]
		\label{def:coherence}
		For each $x\in X$, define
		\[
		\Coh_x
		\coloneqq
		\Bigl(\bigvee_{c\in D_x}p_x^c\Bigr)
		\land
		\bigwedge_{\substack{c,d\in D_x\\ c\neq d}}
		(p_x^c\to\neg p_x^d).
		\]
		The global coherence formula is
		\[
		\Coh_\Theta
		\coloneqq
		\bigwedge_{x\in X}\Coh_x.
		\]
		A propositional coalition model over $P_\Theta$ is \emph{coherent} if every state satisfies $\Coh_\Theta$.
	\end{definition}
	
	When $X=\varnothing$, the empty conjunction defining $\Coh_\Theta$ is understood as $\top$. For $x\in X$, the formula $\Coh_x$ says that at least one value of $x$ holds and that no two distinct values of $x$ hold simultaneously.
	
	\begin{definition}[From assignment models to coherent propositional models]
		\label{def:associated-prop-model}
		Let
		\[
		\mathcal{M}=(S,\{\Sigma_i\}_{i\in\Ag},o,\pi)
		\]
		be a $\VCL$ model over $\Theta$. Its associated propositional coalition model is
		\[
		\mathcal{M}^\ast=(S,\{\Sigma_i\}_{i\in\Ag},o,V),
		\]
		where
		\[
		V(p_x^c)
		\coloneqq
		\{\,s\in S\mid \pi(s)(x)=c\,\}.
		\]
	\end{definition}
	
	\begin{definition}[From coherent propositional models to assignment models]
		\label{def:associated-vcl-model}
		Let
		\[
		\mathcal{N}=(S,\{\Sigma_i\}_{i\in\Ag},o,V)
		\]
		be a coherent propositional coalition model over $P_\Theta$. Define
		\[
		\pi^\dagger:S\to\prod_{x\in X}D_x
		\]
		by letting $\pi^\dagger(s)(x)$ be the unique $c\in D_x$ such that
		\[
		s\in V(p_x^c).
		\]
		The associated $\VCL$ model is
		\[
		\mathcal{N}^\dagger
		=
		(S,\{\Sigma_i\}_{i\in\Ag},o,\pi^\dagger).
		\]
	\end{definition}
	
	\begin{lemma}[Associated models]
		\label{lem:associated-models}
		For every $\VCL$ model $\mathcal{M}$ over $\Theta$, the associated propositional model $\mathcal{M}^\ast$ is coherent. Conversely, if $\mathcal{N}$ is a coherent propositional coalition model over $P_\Theta$, then $\mathcal{N}^\dagger$ is well-defined and unique.
	\end{lemma}
	
	\begin{proof}
		For $\mathcal{M}^\ast$, fix $s\in S$ and $x\in X$. Since $\pi(s)$ is a typed assignment, there is a unique $c^\ast\in D_x$ such that
		\[
		\pi(s)(x)=c^\ast.
		\]
		Hence
		\[
		s\in V(p_x^{c^\ast}).
		\]
		If $d\in D_x$ and $d\neq c^\ast$, then $\pi(s)(x)\neq d$, so
		\[
		s\notin V(p_x^d).
		\]
		Thus $s\models\Coh_x$. Since $s$ and $x$ were arbitrary, every state of $\mathcal{M}^\ast$ satisfies $\Coh_\Theta$.
		
		Conversely, suppose $\mathcal{N}$ is coherent. For every $s\in S$ and every $x\in X$, the exhaustive part of $\Coh_x$ gives at least one $c\in D_x$ with $s\in V(p_x^c)$, while the exclusive part gives at most one such $c$. Hence $\pi^\dagger(s)(x)$ is well-defined. Since this value is uniquely determined for every $s$ and $x$, the map $\pi^\dagger$ is unique.
	\end{proof}
	
	The two constructions are inverse to each other on the intended coherent class.
	
	\begin{proposition}[Near-inverse correspondence]
		\label{prop:near-inverse-correspondence}
		\leavevmode
		\begin{enumerate}[label=\textup{(\roman*)},leftmargin=*]
			\item For every $\VCL$ model $\mathcal{M}$ over $\Theta$,
			\[
			(\mathcal{M}^\ast)^\dagger=\mathcal{M}.
			\]
			\item For every coherent propositional coalition model $\mathcal{N}$ over $P_\Theta$,
			\[
			(\mathcal{N}^\dagger)^\ast=\mathcal{N}.
			\]
		\end{enumerate}
	\end{proposition}
	
	\begin{proof}
		Let
		\[
		\mathcal{M}=(S,\{\Sigma_i\}_{i\in\Ag},o,\pi).
		\]
		The model $(\mathcal{M}^\ast)^\dagger$ has the same state space, action sets, and outcome function as $\mathcal{M}$. Its assignment map sends $s$ and $x$ to the unique $c\in D_x$ such that
		\[
		s\in V(p_x^c).
		\]
		By the definition of $\mathcal{M}^\ast$, this is exactly the unique $c$ satisfying $\pi(s)(x)=c$. Hence the recovered assignment map is $\pi$.
		
		Conversely, let
		\[
		\mathcal{N}=(S,\{\Sigma_i\}_{i\in\Ag},o,V)
		\]
		be coherent. The model $(\mathcal{N}^\dagger)^\ast$ has the same strategic component as $\mathcal{N}$. For every atom $p_x^c\in P_\Theta$,
		\[
		V^\ast(p_x^c)
		=
		\{\,s\in S\mid \pi^\dagger(s)(x)=c\,\}.
		\]
		By the definition of $\pi^\dagger$ and coherence of $\mathcal{N}$, this set is exactly $V(p_x^c)$. Thus the recovered valuation agrees with $V$ on all atoms in $P_\Theta$.
	\end{proof}
	
	\subsection{Truth and Validity Correspondence}
	\label{subsec:truth-validity}
	
	\begin{theorem}[Truth correspondence]
		\label{thm:encoding-equivalence}
		Let $\varphi\in\Lang_{\VCL}(\Theta)$.
		\begin{enumerate}[label=\textup{(\roman*)},leftmargin=*]
			\item For every $\VCL$ model $\mathcal{M}$ over $\Theta$ and every state $s$,
			\[
			\mathcal{M},s\models\varphi
			\quad\text{iff}\quad
			\mathcal{M}^\ast,s\models\tr(\varphi).
			\]
			
			\item For every coherent propositional model $\mathcal{N}$ over $P_\Theta$ and every state $s$,
			\[
			\mathcal{N},s\models\tr(\varphi)
			\quad\text{iff}\quad
			\mathcal{N}^\dagger,s\models\varphi.
			\]
		\end{enumerate}
	\end{theorem}
	
	\begin{proof}
		Both claims are proved by structural induction on $\varphi$. The Boolean cases are immediate.
		
		For atoms, in the first direction,
		\[
		\mathcal{M},s\models(x{=}c)
		\quad\text{iff}\quad
		\pi(s)(x)=c
		\quad\text{iff}\quad
		s\in V(p_x^c)
		\quad\text{iff}\quad
		\mathcal{M}^\ast,s\models p_x^c.
		\]
		The atomic case for coherent propositional models is identical, using the definition of $\pi^\dagger$.
		
		For the modal case, let $\varphi=[C]\psi$. The paired models $\mathcal{M}$ and $\mathcal{M}^\ast$ have the same state space, the same action sets, and the same outcome function. Hence the witnessing $C$-actions and all compatible successor states are identical. Applying the induction hypothesis at each such successor gives
		\[
		\mathcal{M},s\models[C]\psi
		\quad\text{iff}\quad
		\mathcal{M}^\ast,s\models[C]\tr(\psi).
		\]
		The argument for $\mathcal{N}$ and $\mathcal{N}^\dagger$ is the same.
	\end{proof}
	
	\begin{corollary}[Validity correspondence]
		\label{cor:validity-correspondence}
		For every $\varphi\in\Lang_{\VCL}(\Theta)$,
		\[
		\models_{\VCL,\Theta}\varphi
		\quad\text{iff}\quad
		\models_{\mathsf{cohCL},\Theta}\tr(\varphi),
		\]
		where $\models_{\mathsf{cohCL},\Theta}$ denotes validity over coherent propositional coalition models over $P_\Theta$.
	\end{corollary}
	
	\begin{proof}
		Suppose first that $\varphi$ is valid over all $\VCL$ models over $\Theta$. Let $\mathcal{N}$ be a coherent propositional coalition model over $P_\Theta$. By Lemma~\ref{lem:associated-models}, $\mathcal{N}^\dagger$ is a $\VCL$ model over $\Theta$. Hence
		\[
		\mathcal{N}^\dagger,s\models\varphi
		\]
		for every state $s$. By Theorem~\ref{thm:encoding-equivalence},
		\[
		\mathcal{N},s\models\tr(\varphi)
		\]
		for every $s$. Thus $\tr(\varphi)$ is valid over coherent propositional coalition models.
		
		Conversely, suppose $\tr(\varphi)$ is valid over coherent propositional coalition models. Let $\mathcal{M}$ be a $\VCL$ model over $\Theta$. By Lemma~\ref{lem:associated-models}, $\mathcal{M}^\ast$ is coherent. Hence
		\[
		\mathcal{M}^\ast,s\models\tr(\varphi)
		\]
		for every state $s$. By Theorem~\ref{thm:encoding-equivalence},
		\[
		\mathcal{M},s\models\varphi
		\]
		for every $s$. Therefore $\varphi$ is valid over $\VCL$ models.
	\end{proof}
	
	Thus, for a fixed finite typed signature, $\VCL$ is validity-equivalent to propositional Coalition Logic over coherent valuations. The equivalence is not merely a translation of formulas; it is supported by a near-inverse correspondence between assignment models and coherent propositional models over the same strategic component.
	
	\subsection{Hilbert System}
	\label{subsec:hilbert-system}
	
	We now give a Hilbert-style system for $\VCL$. It consists of the usual one-step principles of Coalition Logic together with value-coherence axioms.
	
	\begin{definition}[The proof system $\VCL_\Theta$]
		\label{def:vcl-proof-system}
		The system $\VCL_\Theta$ has the following axiom schemata and rules.
		
		\paragraph{Propositional base.}
		\begin{enumerate}[label=\textup{(P\arabic*)},leftmargin=*]
			\item All substitution instances of classical propositional tautologies.
		\end{enumerate}
		
		\paragraph{Coalitional core.}
		For all $C,D\subseteq\Ag$:
		\begin{enumerate}[label=\textup{(C\arabic*)},leftmargin=*]
			\item $[C](\varphi\land\psi)\to[C]\varphi$.
			\item $\neg[C]\bot$.
			\item $[C]\top$.
			\item $[C]\varphi\land[D]\psi\to[C\cup D](\varphi\land\psi)$, whenever $C\cap D=\varnothing$.
			\item $\neg[\varnothing]\neg\varphi\to[\Ag]\varphi$.
		\end{enumerate}
		
		\paragraph{Value coherence.}
		\begin{enumerate}[label=\textup{(V\arabic*)},leftmargin=*]
			\item $\displaystyle\bigvee_{c\in D_x}(x{=}c)$ for each $x\in X$.
			\item $(x{=}c)\to\neg(x{=}d)$ for each $x\in X$ and all distinct $c,d\in D_x$.
		\end{enumerate}
		
		\paragraph{Rules.}
		\begin{enumerate}[label=\textup{(R\arabic*)},leftmargin=*]
			\item \textup{(MP)} From $\varphi$ and $\varphi\to\psi$, infer $\psi$.
			\item \textup{(RE)} From $\vdash\varphi\leftrightarrow\psi$, infer
			\[
			\vdash [C]\varphi\leftrightarrow[C]\psi
			\]
			for every $C\subseteq\Ag$.
		\end{enumerate}
	\end{definition}
	
	Axioms \textup{(C1)}--\textup{(C5)} are the standard one-step coalition principles~\cite{Pauly02,Pauly02Comp}. Axioms \textup{(V1)} and \textup{(V2)} are the proof-theoretic counterparts of assignment coherence.
	
	For comparison, let $\mathsf{CL}_{\Coh,\Theta}$ be the propositional companion system over $P_\Theta$ obtained by adding to the usual Hilbert system for classical Coalition Logic the coherence schemata
	\[
	\bigvee_{c\in D_x}p_x^c
	\qquad
	(x\in X),
	\]
	and
	\[
	p_x^c\to\neg p_x^d
	\qquad
	(x\in X,\ c,d\in D_x,\ c\neq d).
	\]
	
	\begin{lemma}[Syntactic correspondence]
		\label{lem:syntactic-correspondence}
		For every $\varphi\in\Lang_{\VCL}(\Theta)$,
		\[
		\vdash_{\VCL_\Theta}\varphi
		\quad\text{iff}\quad
		\vdash_{\mathsf{CL}_{\Coh,\Theta}}\tr(\varphi).
		\]
	\end{lemma}
	
	\begin{proof}
		For the left-to-right direction, translate each line of a $\VCL_\Theta$ derivation by $\tr$. Propositional tautologies, coalition axioms, and the rules are preserved by the homomorphic definition of $\tr$, and \textup{(V1)}--\textup{(V2)} become exactly the coherence schemata of $\mathsf{CL}_{\Coh,\Theta}$.
		
		For the right-to-left direction, define
		\[
		b:\Lang_{\CL}(P_\Theta)\to\Lang_{\VCL}(\Theta)
		\]
		by
		\[
		b(p_x^c)=(x{=}c),
		\]
		and by commuting with Boolean connectives and coalition modalities. Applying $b$ to each line of a $\mathsf{CL}_{\Coh,\Theta}$ derivation gives a $\VCL_\Theta$ derivation: propositional tautologies, coalition axioms, rules, and coherence axioms are all preserved. Since
		\[
		b(\tr(\varphi))=\varphi,
		\]
		the claim follows.
	\end{proof}
	
	\subsection{Soundness and Completeness}
	\label{subsec:soundness-completeness}
	
	\begin{theorem}[Soundness]
		\label{thm:soundness}
		For every $\varphi\in\Lang_{\VCL}(\Theta)$,
		\[
		\vdash_{\VCL_\Theta}\varphi
		\quad\Longrightarrow\quad
		\models_{\VCL,\Theta}\varphi.
		\]
	\end{theorem}
	
	\begin{proof}
		Propositional tautologies are valid, and modus ponens preserves validity. Rule \textup{(RE)} is sound because logically equivalent formulas define the same set of successor states.
		
		The coalitional axioms are valid by the standard one-step argument. Axiom \textup{(C1)} follows because any action guaranteeing $\varphi\land\psi$ also guarantees $\varphi$. Axiom \textup{(C2)} follows from non-empty action sets and totality of the outcome function, since no coalition can force $\bot$. Axiom \textup{(C3)} is immediate. Axiom \textup{(C4)} follows by combining witnessing actions for disjoint coalitions. For \textup{(C5)}, observe that $\neg[\varnothing]\neg\varphi$ means that some complete action profile leads to a successor satisfying $\varphi$; the grand coalition can choose precisely such a complete profile, and hence $[\Ag]\varphi$ holds.
		
		Finally, \textup{(V1)} and \textup{(V2)} are valid by Theorem~\ref{thm:assignment-coherence}.
	\end{proof}
	
	\begin{lemma}[Completeness of the coherent companion]
		\label{lem:cl-coh-completeness}
		For every $\chi\in\Lang_{\CL}(P_\Theta)$,
		\[
		\models_{\mathsf{cohCL},\Theta}\chi
		\quad\text{iff}\quad
		\vdash_{\mathsf{CL}_{\Coh,\Theta}}\chi .
		\]
	\end{lemma}
	
	\begin{proof}
		Soundness is immediate from the soundness of Coalition Logic and the definition of coherent models: the added coherence schemata are valid at every state of every coherent propositional coalition model.
		
		For completeness, argue contrapositively. Suppose
		\[
		\nvdash_{\mathsf{CL}_{\Coh,\Theta}}\chi .
		\]
		Then $\{\neg\chi\}$ is $\mathsf{CL}_{\Coh,\Theta}$-consistent and can be extended to a maximal $\mathsf{CL}_{\Coh,\Theta}$-consistent set $\Gamma$.
		
		Use the standard canonical construction for classical Coalition Logic, but with maximal consistent sets taken relative to the extended system $\mathsf{CL}_{\Coh,\Theta}$~\cite{Pauly02,Pauly02Comp}. The modal part of the construction is unchanged, since the added axioms are purely propositional coherence principles. Moreover, every theorem of $\mathsf{CL}_{\Coh,\Theta}$ belongs to every canonical state. In particular, every canonical state contains all instances of the coherence schemata. Hence every canonical state satisfies
		\[
		\Coh_\Theta.
		\]
		Therefore the canonical propositional coalition model is coherent.
		
		By the usual truth lemma for the canonical model, the state $\Gamma$ satisfies $\neg\chi$ and therefore falsifies $\chi$. Thus there is a coherent propositional coalition model in which $\chi$ is false.
		
		If one presents the standard completeness theorem for Coalition Logic first in effectivity-function semantics, the final passage to explicit game forms is obtained by the usual representation of playable, in particular truly playable, effectivity functions by strategic game forms~\cite{GJT13}. This representation preserves the state space and the valuation on $P_\Theta$, so the coherence condition is preserved. Hence $\chi$ is not valid over coherent explicit one-step propositional coalition models. This proves the contrapositive of completeness.
	\end{proof}
	
	\begin{theorem}[Completeness]
		\label{thm:completeness}
		For every $\varphi\in\Lang_{\VCL}(\Theta)$,
		\[
		\models_{\VCL,\Theta}\varphi
		\quad\Longrightarrow\quad
		\vdash_{\VCL_\Theta}\varphi .
		\]
	\end{theorem}
	
	\begin{proof}
		Assume
		\[
		\models_{\VCL,\Theta}\varphi .
		\]
		By Corollary~\ref{cor:validity-correspondence},
		\[
		\models_{\mathsf{cohCL},\Theta}\tr(\varphi).
		\]
		By Lemma~\ref{lem:cl-coh-completeness},
		\[
		\vdash_{\mathsf{CL}_{\Coh,\Theta}}\tr(\varphi).
		\]
		By Lemma~\ref{lem:syntactic-correspondence},
		\[
		\vdash_{\VCL_\Theta}\varphi .
		\]
	\end{proof}
	
	\begin{corollary}[Soundness and completeness]
		\label{cor:sound-complete}
		For every $\varphi\in\Lang_{\VCL}(\Theta)$,
		\[
		\vdash_{\VCL_\Theta}\varphi
		\quad\text{iff}\quad
		\models_{\VCL,\Theta}\varphi .
		\]
	\end{corollary}
	
	\begin{remark}[Conservativity]
		The axiomatisation confirms the conservative status of $\VCL$ at the level of the coalition modality. The modal principles are exactly the usual one-step principles of Coalition Logic. The only additional axioms are the finite-domain value-coherence principles induced by typed assignments. Thus $\VCL$ does not extend the strategic semantics of $\CL$; it isolates the coherent finite-value fragment of its propositional presentation and treats that fragment as typed semantic structure.
	\end{remark}

	\section{Value Quotients and Projected Effectivity}
	\label{sec:value-quotients}
	
	The assignment semantics of $\VCL$ allows ordinary one-step coalitional ability to be projected onto the value domain of a single variable. For each model, state, and variable, the underlying state-valued game form induces a value-level game form whose outcomes are values of that variable. The value regions enforceable by a coalition are exactly the ordinary effectivity sets of this projected game form. Equivalently, they are generated by the value ranges produced by the coalition's concrete actions.
	
	Throughout this section, fix a typed signature
	\[
	\Theta=(\Ag,X,\{D_x\}_{x\in X})
	\]
	and a $\VCL$ model
	\[
	\mathcal{M}=(S,\{\Sigma_i\}_{i\in\Ag},o,\pi)
	\]
	over $\Theta$. Recall that
	\[
	\Sigma_C=\prod_{i\in C}\Sigma_i
	\]
	for each coalition $C\subseteq\Ag$.
	
	\subsection{Value Regions and Projected Effectivity}
	\label{subsec:projected-effectivity}
	
	For a variable $x\in X$ and a set of values $A\subseteq D_x$, write
	\[
	(x\in A)
	\]
	as an abbreviation for
	\[
	\bigvee_{c\in A}(x{=}c).
	\]
	If $A=\varnothing$, this abbreviation is understood as $\bot$. Thus, for every state $s\in S$,
	\[
	\mathcal{M},s\models (x\in A)
	\quad\text{iff}\quad
	\pi(s)(x)\in A.
	\]
	In particular, $(x\in D_x)$ is valid and $(x\in\varnothing)$ is equivalent to $\bot$.
	
	The formula $[C](x\in A)$ says that coalition $C$ can force the next value of $x$ to lie in the value region $A$. This gives the following projected effectivity family.
	
	\begin{definition}[Projected value-effectivity]
		\label{def:projected-value-effectivity}
		Let $s\in S$, $x\in X$, and $C\subseteq\Ag$. The \emph{projected value-effectivity family} of $C$ for $x$ at $s$ is
		\[
		\mathcal{N}_x^{\mathcal{M}}(s,C)
		\coloneqq
		\{\,A\subseteq D_x
		\mid
		\mathcal{M},s\models [C](x\in A)
		\,\}.
		\]
	\end{definition}
	
	Unfolding the semantics, we have
	\[
	A\in\mathcal{N}_x^{\mathcal{M}}(s,C)
	\]
	iff there exists
	\[
	\alpha_C\in\Sigma_C
	\]
	such that, for every
	\[
	\beta_{\overline C}\in\Sigma_{\overline C},
	\]
	we have
	\[
	\pi\bigl(o(s,\alpha_C\sqcup\beta_{\overline C})\bigr)(x)\in A.
	\]
	Thus $\mathcal{N}_x^{\mathcal{M}}(s,C)$ collects precisely those regions of $D_x$ to which coalition $C$ can restrict the next value of $x$.
	
	It is important that $\mathcal{N}_x^{\mathcal{M}}(s,C)$ is not introduced as a new primitive effectivity semantics. It is obtained by applying the ordinary one-step coalition modality to value-region formulas and then reading the result inside the finite Boolean lattice $\mathcal{P}(D_x)$.
	
	\subsection{Value-Quotient Game Forms}
	\label{subsec:value-quotient-game-forms}
	
	For a fixed state $s$ and variable $x$, the original outcome function can be composed with the value projection
	\[
	q_x:S\to D_x,
	\qquad
	q_x(t)=\pi(t)(x).
	\]
	Equivalently, define an equivalence relation on states by
	\[
	t\equiv_x u
	\quad\text{iff}\quad
	\pi(t)(x)=\pi(u)(x).
	\]
	The value projection $q_x$ identifies successor states that agree on the value of $x$. Factoring the state-valued outcome map through this projection gives a canonical value-level game form.
	
	\begin{definition}[Value-quotient game form]
		\label{def:value-quotient-game-form}
		Let $s\in S$ and $x\in X$. The \emph{$x$-value quotient game form at $s$} is
		\[
		\mathcal{G}_x^s
		=
		(D_x,\{\Sigma_i\}_{i\in\Ag},o_x^s),
		\]
		where
		\[
		o_x^s:\Sigma_{\Ag}\to D_x
		\]
		is defined by
		\[
		o_x^s(\gamma)
		\coloneqq
		\pi(o(s,\gamma))(x).
		\]
	\end{definition}
	
	The term ``quotient'' refers to factoring the state-valued outcome map through the $x$-value equivalence relation. No quotient of the action space is involved. The game form $\mathcal{G}_x^s$ has the same agents and action sets as the original one-step model, but its outcomes are values of $x$ rather than states. Some values in $D_x$ may fail to be realised at $s$ by any action profile; they nevertheless remain part of the typed outcome domain.
	
	For comparison with the standard effectivity-function presentation of Coalition Logic~\cite{Pauly02,GJT13}, recall the ordinary game-form notion. If
	\[
	\mathcal{G}=(O,\{\Sigma_i\}_{i\in\Ag},g)
	\]
	is a game form with outcome set $O$, its induced effectivity family for coalition $C$ is
	\[
	E_{\mathcal{G}}(C)
	=
	\{\,A\subseteq O
	\mid
	\exists\,\alpha_C\in\Sigma_C\
	\forall\,\beta_{\overline C}\in\Sigma_{\overline C}:
	g(\alpha_C\sqcup\beta_{\overline C})\in A
	\,\}.
	\]
	
	\begin{theorem}[Value-quotient representation]
		\label{thm:value-quotient-representation}
		Let $s\in S$ and $x\in X$. For every coalition $C\subseteq\Ag$,
		\[
		\mathcal{N}_x^{\mathcal{M}}(s,C)
		=
		E_{\mathcal{G}_x^s}(C).
		\]
		Equivalently,
		\[
		\mathcal{N}_x^{\mathcal{M}}(s,C)
		=
		\left\{
		A\subseteq D_x
		\mid
		\exists\,\alpha_C\in\Sigma_C\;
		\forall\,\beta_{\overline C}\in\Sigma_{\overline C}:
		o_x^s(\alpha_C\sqcup\beta_{\overline C})\in A
		\right\}.
		\]
	\end{theorem}
	
	\begin{proof}
		Let $A\subseteq D_x$. Then:
		\begin{align*}
			A\in \mathcal{N}_x^{\mathcal{M}}(s,C)
			&\quad\text{iff}\quad
			\mathcal{M},s\models [C](x\in A)\\
			&\quad\text{iff}\quad
			\exists\,\alpha_C\in\Sigma_C\;
			\forall\,\beta_{\overline C}\in\Sigma_{\overline C}:\\
			&\hspace{2.8cm}
			\mathcal{M},o(s,\alpha_C\sqcup\beta_{\overline C})\models (x\in A)\\
			&\quad\text{iff}\quad
			\exists\,\alpha_C\in\Sigma_C\;
			\forall\,\beta_{\overline C}\in\Sigma_{\overline C}:\\
			&\hspace{2.8cm}
			\pi(o(s,\alpha_C\sqcup\beta_{\overline C}))(x)\in A\\
			&\quad\text{iff}\quad
			\exists\,\alpha_C\in\Sigma_C\;
			\forall\,\beta_{\overline C}\in\Sigma_{\overline C}:\\
			&\hspace{2.8cm}
			o_x^s(\alpha_C\sqcup\beta_{\overline C})\in A\\
			&\quad\text{iff}\quad
			A\in E_{\mathcal{G}_x^s}(C).
		\end{align*}
		The first equivalence is Definition~\ref{def:projected-value-effectivity}; the second is the truth clause for $[C]$; the third is the semantics of value regions; the fourth is Definition~\ref{def:value-quotient-game-form}; and the last is the ordinary game-form definition of effectivity.
	\end{proof}
	
	Thus value-level control is not an additional strategic mechanism. It is ordinary one-step coalitional ability factored through the canonical value projection determined by the typed assignment map.
	
	\subsection{Strategic Value-Range Hypergraphs}
	\label{subsec:value-range-hypergraphs}
	
	The quotient perspective gives a concrete action-based representation of projected effectivity. Each action of a coalition determines the set of values that may still occur after the complementary coalition responds.
	
	\begin{definition}[Strategic value range]
		\label{def:strategic-value-range}
		Let $s\in S$, $x\in X$, $C\subseteq\Ag$, and $\alpha_C\in\Sigma_C$. The \emph{$x$-value range} of $\alpha_C$ at $s$ is
		\[
		R_x^C(s,\alpha_C)
		\coloneqq
		\{
		\pi(o(s,\alpha_C\sqcup\beta_{\overline C}))(x)
		\mid
		\beta_{\overline C}\in\Sigma_{\overline C}
		\}.
		\]
	\end{definition}
	
	Thus $R_x^C(s,\alpha_C)$ is the set of all possible next values of $x$ compatible with coalition $C$ choosing $\alpha_C$ and the complementary coalition choosing arbitrarily. Since all action sets are non-empty, every such range is non-empty.
	
	\begin{definition}[Strategic value-range hypergraph]
		\label{def:strategic-value-range-hypergraph}
		Let $s\in S$, $x\in X$, and $C\subseteq\Ag$. The \emph{strategic value-range hypergraph} of $C$ for $x$ at $s$ is
		\[
		\mathcal{R}_x^{\mathcal{M}}(s,C)
		\coloneqq
		\{
		R_x^C(s,\alpha_C)
		\mid
		\alpha_C\in\Sigma_C
		\}.
		\]
	\end{definition}
	
	The members of $\mathcal{R}_x^{\mathcal{M}}(s,C)$ are the value ranges generated by concrete $C$-actions. Smaller ranges correspond to more precise value control.
	
	For a family $\mathcal{H}\subseteq\mathcal{P}(D_x)$, write
	\[
	\uparrow\mathcal{H}
	\coloneqq
	\{\,A\subseteq D_x
	\mid
	\exists H\in\mathcal{H}\text{ such that }H\subseteq A
	\,\}
	\]
	for its upward closure in the Boolean lattice $\mathcal{P}(D_x)$.
	
	\begin{theorem}[Operational hypergraph representation]
		\label{thm:operational-hypergraph-representation}
		For every state $s\in S$, variable $x\in X$, and coalition $C\subseteq\Ag$,
		\[
		\mathcal{N}_x^{\mathcal{M}}(s,C)
		=
		\uparrow\mathcal{R}_x^{\mathcal{M}}(s,C).
		\]
		Equivalently,
		\[
		\mathcal{N}_x^{\mathcal{M}}(s,C)
		=
		\left\{
		A\subseteq D_x
		\mid
		\exists R\in\mathcal{R}_x^{\mathcal{M}}(s,C)
		\text{ such that }
		R\subseteq A
		\right\}.
		\]
	\end{theorem}
	
	\begin{proof}
		Let $A\subseteq D_x$. By Definition~\ref{def:projected-value-effectivity},
		\[
		A\in\mathcal{N}_x^{\mathcal{M}}(s,C)
		\]
		iff coalition $C$ has an action $\alpha_C$ such that every compatible response of $\overline C$ leads to a successor whose $x$-value belongs to $A$. By Definition~\ref{def:strategic-value-range}, this is precisely the condition
		\[
		R_x^C(s,\alpha_C)\subseteq A.
		\]
		Equivalently, there is some
		\[
		R\in\mathcal{R}_x^{\mathcal{M}}(s,C)
		\]
		with $R\subseteq A$. Hence
		\[
		\mathcal{N}_x^{\mathcal{M}}(s,C)
		=
		\uparrow\mathcal{R}_x^{\mathcal{M}}(s,C).
		\]
	\end{proof}
	
	\begin{definition}[Minimal strategic generators]
		\label{def:minimal-strategic-generators}
		For $s\in S$, $x\in X$, and $C\subseteq\Ag$, define
		\[
		\mathsf{Gen}_x^{\mathcal{M}}(s,C)
		\coloneqq
		\operatorname{Min}_{\subseteq}
		\mathcal{R}_x^{\mathcal{M}}(s,C),
		\]
		the family of inclusion-minimal strategic value ranges generated by $C$-actions.
	\end{definition}
	
	\begin{corollary}[Minimal strategic ranges]
		\label{cor:minimal-strategic-ranges}
		For every state $s\in S$, variable $x\in X$, and coalition $C\subseteq\Ag$,
		\[
		\operatorname{Min}_{\subseteq}
		\mathcal{N}_x^{\mathcal{M}}(s,C)
		=
		\mathsf{Gen}_x^{\mathcal{M}}(s,C).
		\]
		Consequently, $\mathsf{Gen}_x^{\mathcal{M}}(s,C)$ is an antichain and
		\[
		\mathcal{N}_x^{\mathcal{M}}(s,C)
		=
		\uparrow \mathsf{Gen}_x^{\mathcal{M}}(s,C).
		\]
	\end{corollary}
	
	\begin{proof}
		The family $\mathcal{R}_x^{\mathcal{M}}(s,C)$ is non-empty because $\Sigma_C$ is non-empty. Since $D_x$ is finite, every non-empty subfamily of $\mathcal{P}(D_x)$ has inclusion-minimal members.
		
		For any non-empty family $\mathcal{H}\subseteq\mathcal{P}(D_x)$,
		\[
		\operatorname{Min}_{\subseteq}(\uparrow\mathcal{H})
		=
		\operatorname{Min}_{\subseteq}\mathcal{H}.
		\]
		Indeed, if $H$ is minimal in $\mathcal{H}$ and $A\in\uparrow\mathcal{H}$ with $A\subseteq H$, then some $H'\in\mathcal{H}$ satisfies $H'\subseteq A\subseteq H$. By minimality of $H$, we get $H'=H$, and hence $A=H$. Thus $H$ is minimal in $\uparrow\mathcal{H}$.
		
		Conversely, if $A$ is minimal in $\uparrow\mathcal{H}$, then some $H\in\mathcal{H}$ satisfies $H\subseteq A$. Since $H\in\uparrow\mathcal{H}$, minimality of $A$ gives $A=H$. Moreover, if some $H'\in\mathcal{H}$ satisfied $H'\subsetneq H$, then $H'\in\uparrow\mathcal{H}$ and $H'\subsetneq A$, contradicting minimality of $A$. Hence $H$ is minimal in $\mathcal{H}$.
		
		Applying this fact to
		\[
		\mathcal{H}=\mathcal{R}_x^{\mathcal{M}}(s,C)
		\]
		and using Theorem~\ref{thm:operational-hypergraph-representation} gives
		\[
		\operatorname{Min}_{\subseteq}
		\mathcal{N}_x^{\mathcal{M}}(s,C)
		=
		\operatorname{Min}_{\subseteq}
		\mathcal{R}_x^{\mathcal{M}}(s,C)
		=
		\mathsf{Gen}_x^{\mathcal{M}}(s,C).
		\]
		Inclusion-minimal members form an antichain. Since $\mathcal{N}_x^{\mathcal{M}}(s,C)$ is the upward closure of $\mathcal{R}_x^{\mathcal{M}}(s,C)$, and since removing non-minimal generators does not change an upward closure, we obtain
		\[
		\mathcal{N}_x^{\mathcal{M}}(s,C)
		=
		\uparrow \mathsf{Gen}_x^{\mathcal{M}}(s,C).
		\]
	\end{proof}
	
	\begin{example}
		\label{ex:value-range-hypergraph}
		Let $\Ag=\{1,2\}$, let $D_x=\{a,b,c\}$, and suppose that at state $s$ the induced quotient game form $\mathcal{G}_x^s$ is given by
		\[
		\begin{array}{c|cc}
			& \beta_1 & \beta_2\\
			\hline
			\alpha_1 & a & b\\
			\alpha_2 & b & c
		\end{array}
		\]
		where $\alpha_1,\alpha_2$ are actions of agent $1$ and $\beta_1,\beta_2$ are actions of agent $2$. For $C=\{1\}$,
		\[
		R_x^C(s,\alpha_1)=\{a,b\},
		\qquad
		R_x^C(s,\alpha_2)=\{b,c\}.
		\]
		Hence
		\[
		\mathcal{R}_x^{\mathcal{M}}(s,C)
		=
		\{\{a,b\},\{b,c\}\},
		\]
		and
		\[
		\mathcal{N}_x^{\mathcal{M}}(s,C)
		=
		\uparrow\{\{a,b\},\{b,c\}\}.
		\]
		Thus coalition $\{1\}$ can force $x$ into $\{a,b\}$ or into $\{b,c\}$, but it cannot force any singleton value. The minimal enforceable value regions form the antichain
		\[
		\mathsf{Gen}_x^{\mathcal{M}}(s,\{1\})
		=
		\{\{a,b\},\{b,c\}\}.
		\]
	\end{example}
	
	\begin{remark}[Operational meaning]
		The theorem and corollary give an operational interpretation of value control. Coalition $C$ can enforce a value region $A$ exactly when it has an action whose possible $x$-values are all contained in $A$. Thus the minimally enforceable value regions are not abstract semantic artefacts; they are precisely the inclusion-minimal strategic value ranges generated by concrete coalition actions.
	\end{remark}

	\section{Strategic Exclusion, Duality, and Indeterminacy}
	\label{sec:structural-results}
	
	This section derives structural consequences of projected value-effectivity. No new strategic modality is introduced. All results follow from ordinary one-step game-form effectivity after projection to a finite value domain. The point is that, once the atomic layer is organised into typed value domains, familiar effectivity-theoretic constraints become value-level constraints inside the Boolean lattice $\mathcal{P}(D_x)$.
	
	Throughout this section, let
	\[
	\Theta=(\Ag,X,\{D_x\}_{x\in X})
	\]
	be a finite typed signature, and let
	\[
	\mathcal{M}=(S,\{\Sigma_i\}_{i\in\Ag},o,\pi)
	\]
	be a $\VCL$ model over $\Theta$. As before, for $C\subseteq\Ag$ we write
	\[
	\Sigma_C=\prod_{i\in C}\Sigma_i .
	\]
	
	\subsection{Basic Laws of Projected Value-Effectivity}
	\label{subsec:basic-laws}
	
	Recall that, for $s\in S$, $x\in X$, and $C\subseteq\Ag$,
	\[
	\mathcal{N}_x^{\mathcal{M}}(s,C)
	=
	\{\,A\subseteq D_x
	\mid
	\mathcal{M},s\models [C](x\in A)
	\,\}.
	\]
	By Theorem~\ref{thm:value-quotient-representation}, this is exactly the ordinary effectivity family of coalition $C$ in the value quotient game form $\mathcal{G}_x^s$. The following laws are therefore inherited from standard one-step game-form effectivity.
	
	\begin{theorem}[Basic laws of projected value-effectivity]
		\label{thm:projected-value-effectivity}
		Let $s\in S$ and $x\in X$. For all coalitions $C,E\subseteq\Ag$:
		\begin{enumerate}[label=\textup{(\roman*)},leftmargin=*]
			\item \textup{(\emph{Value monotonicity})}
			If
			\[
			A\in\mathcal{N}_x^{\mathcal{M}}(s,C)
			\quad\text{and}\quad
			A\subseteq B\subseteq D_x,
			\]
			then
			\[
			B\in\mathcal{N}_x^{\mathcal{M}}(s,C).
			\]
			
			\item \textup{(\emph{Coalition monotonicity})}
			If $C\subseteq E$, then
			\[
			\mathcal{N}_x^{\mathcal{M}}(s,C)
			\subseteq
			\mathcal{N}_x^{\mathcal{M}}(s,E).
			\]
			
			\item \textup{(\emph{Liveness and safety})}
			\[
			D_x\in\mathcal{N}_x^{\mathcal{M}}(s,C)
			\qquad\text{and}\qquad
			\varnothing\notin\mathcal{N}_x^{\mathcal{M}}(s,C).
			\]
			
			\item \textup{(\emph{Meet-superadditivity})}
			If $C\cap E=\varnothing$,
			\[
			A\in\mathcal{N}_x^{\mathcal{M}}(s,C),
			\qquad
			B\in\mathcal{N}_x^{\mathcal{M}}(s,E),
			\]
			then
			\[
			A\cap B
			\in
			\mathcal{N}_x^{\mathcal{M}}(s,C\cup E).
			\]
		\end{enumerate}
	\end{theorem}
	
	\begin{proof}
		For \textup{(i)}, suppose $A\in\mathcal{N}_x^{\mathcal{M}}(s,C)$ and $A\subseteq B\subseteq D_x$. Let $\alpha_C\in\Sigma_C$ witness that $C$ can force $x\in A$. Then every compatible successor has its $x$-value in $A$, hence also in $B$. The same action therefore witnesses $B\in\mathcal{N}_x^{\mathcal{M}}(s,C)$.
		
		For \textup{(ii)}, suppose $C\subseteq E$ and
		\[
		A\in\mathcal{N}_x^{\mathcal{M}}(s,C).
		\]
		Let $\alpha_C\in\Sigma_C$ witness this. Since all individual action sets are non-empty, choose arbitrary actions for the agents in $E\setminus C$ and extend $\alpha_C$ to an $E$-action $\alpha_E\in\Sigma_E$. Every response of $\overline E$, together with $\alpha_E$, determines a complete profile extending the original $C$-action $\alpha_C$. Hence the guarantee of $A$ is preserved, and so
		\[
		A\in\mathcal{N}_x^{\mathcal{M}}(s,E).
		\]
		
		For \textup{(iii)}, every successor state assigns to $x$ some value in $D_x$. Hence every coalition action guarantees $x\in D_x$, so
		\[
		D_x\in\mathcal{N}_x^{\mathcal{M}}(s,C).
		\]
		On the other hand, $\Sigma_{\overline C}$ is non-empty and the outcome function is total. Thus every $C$-action has at least one compatible complete profile and hence at least one successor. Since no successor can satisfy $x\in\varnothing$, no $C$-action can guarantee $\varnothing$.
		
		For \textup{(iv)}, let $\alpha_C\in\Sigma_C$ witness $A\in\mathcal{N}_x^{\mathcal{M}}(s,C)$, and let $\alpha_E\in\Sigma_E$ witness $B\in\mathcal{N}_x^{\mathcal{M}}(s,E)$. Since $C\cap E=\varnothing$, the combined action
		\[
		\alpha_C\sqcup\alpha_E\in\Sigma_{C\cup E}
		\]
		is available to $C\cup E$. For every response of $\overline{C\cup E}$, the resulting complete profile extends both $\alpha_C$ and $\alpha_E$. The successor therefore has its $x$-value in $A$ and in $B$, hence in $A\cap B$. Thus
		\[
		A\cap B\in\mathcal{N}_x^{\mathcal{M}}(s,C\cup E).
		\]
	\end{proof}
	
	\subsection{Set-Valued Strategic Exclusion}
	\label{subsec:set-valued-exclusion}
	
	Meet-superadditivity implies that disjoint coalitions cannot force disjoint regions of the same value domain. This is not an additional incompatibility axiom; it is ordinary coalition superadditivity applied to mutually exclusive value regions.
	
	\begin{theorem}[Set-valued strategic exclusion]
		\label{thm:set-valued-strategic-exclusion}
		Let $x\in X$, let $A,B\subseteq D_x$, and let $C,E\subseteq\Ag$ be disjoint coalitions. If
		\[
		A\cap B=\varnothing,
		\]
		then
		\[
		\models_{\VCL,\Theta}
		[C](x\in A)\to\neg[E](x\in B).
		\]
	\end{theorem}
	
	\begin{proof}
		Let $\mathcal{M}$ be any $\VCL$ model over $\Theta$ and let $s\in S$. Suppose, for contradiction, that
		\[
		\mathcal{M},s\models [C](x\in A)\land[E](x\in B).
		\]
		Then
		\[
		A\in\mathcal{N}_x^{\mathcal{M}}(s,C),
		\qquad
		B\in\mathcal{N}_x^{\mathcal{M}}(s,E).
		\]
		Since $C\cap E=\varnothing$, Theorem~\ref{thm:projected-value-effectivity}\textup{(iv)} gives
		\[
		A\cap B\in\mathcal{N}_x^{\mathcal{M}}(s,C\cup E).
		\]
		But $A\cap B=\varnothing$, contradicting Theorem~\ref{thm:projected-value-effectivity}\textup{(iii)}. Therefore
		\[
		\mathcal{M},s\models [C](x\in A)\to\neg[E](x\in B).
		\]
		Since $\mathcal{M}$ and $s$ were arbitrary, the formula is valid.
	\end{proof}
	
	\begin{corollary}[Singleton strategic exclusion]
		\label{cor:strategic-exclusion}
		If $c,d\in D_x$ with $c\neq d$ and $C\cap E=\varnothing$, then
		\[
		\models_{\VCL,\Theta}
		[C](x{=}c)\to\neg[E](x{=}d).
		\]
	\end{corollary}
	
	\begin{proof}
		Apply Theorem~\ref{thm:set-valued-strategic-exclusion} to
		\[
		A=\{c\},
		\qquad
		B=\{d\}.
		\]
	\end{proof}
	
	Thus exclusion is a value-level shadow of superadditivity: if two disjoint coalitions could force disjoint value regions of the same variable, their union would force the empty region, which is impossible.
	
	\subsection{Transversal Polarity}
	\label{subsec:transversal-polarity}
	
	Projected value-effectivity families are upward-closed families in the Boolean lattice $\mathcal{P}(D_x)$. The incompatibility of disjoint enforceable value regions can therefore be expressed as a transversal condition.
	
	\begin{definition}[Transversal dual]
		\label{def:transversal-dual}
		Let $D$ be a finite set and let
		\[
		\mathcal{F}\subseteq\mathcal{P}(D).
		\]
		The \emph{transversal dual} of $\mathcal{F}$ is
		\[
		\mathcal{F}^{\#}
		\coloneqq
		\{\,B\subseteq D
		\mid
		\text{for every }A\in\mathcal{F},\ A\cap B\neq\varnothing
		\,\}.
		\]
	\end{definition}
	
	\begin{theorem}[Strategic transversal polarity]
		\label{thm:strategic-transversal-polarity}
		Let $s\in S$ and $x\in X$. If $C,E\subseteq\Ag$ are disjoint, then
		\[
		\mathcal{N}_x^{\mathcal{M}}(s,E)
		\subseteq
		\bigl(\mathcal{N}_x^{\mathcal{M}}(s,C)\bigr)^{\#},
		\]
		and symmetrically,
		\[
		\mathcal{N}_x^{\mathcal{M}}(s,C)
		\subseteq
		\bigl(\mathcal{N}_x^{\mathcal{M}}(s,E)\bigr)^{\#}.
		\]
	\end{theorem}
	
	\begin{proof}
		We prove the first inclusion; the second is symmetric. Let
		\[
		B\in\mathcal{N}_x^{\mathcal{M}}(s,E)
		\]
		and let
		\[
		A\in\mathcal{N}_x^{\mathcal{M}}(s,C).
		\]
		Since $C\cap E=\varnothing$, Theorem~\ref{thm:projected-value-effectivity}\textup{(iv)} gives
		\[
		A\cap B\in\mathcal{N}_x^{\mathcal{M}}(s,C\cup E).
		\]
		By Theorem~\ref{thm:projected-value-effectivity}\textup{(iii)},
		\[
		\varnothing\notin\mathcal{N}_x^{\mathcal{M}}(s,C\cup E).
		\]
		Hence
		\[
		A\cap B\neq\varnothing.
		\]
		Since $A$ was arbitrary, $B$ intersects every member of $\mathcal{N}_x^{\mathcal{M}}(s,C)$. Thus
		\[
		B\in
		\bigl(\mathcal{N}_x^{\mathcal{M}}(s,C)\bigr)^{\#}.
		\]
	\end{proof}
	
	The operation $(\cdot)^{\#}$ is antitone:
	\[
	\mathcal{F}\subseteq\mathcal{G}
	\quad\Longrightarrow\quad
	\mathcal{G}^{\#}\subseteq\mathcal{F}^{\#}.
	\]
	Thus projected value-effectivity families of disjoint coalitions are constrained by a hypergraph transversal relation.
	
	The same constraint can be stated directly at the level of minimal strategic generators.
	
	\begin{corollary}[Generator-level transversal constraint]
		\label{cor:generator-transversal-constraint}
		Let $s\in S$ and $x\in X$. If $C,E\subseteq\Ag$ are disjoint, then for all
		\[
		G\in \mathsf{Gen}_x^{\mathcal{M}}(s,C)
		\qquad\text{and}\qquad
		H\in \mathsf{Gen}_x^{\mathcal{M}}(s,E),
		\]
		we have
		\[
		G\cap H\neq\varnothing.
		\]
	\end{corollary}
	
	\begin{proof}
		By Corollary~\ref{cor:minimal-strategic-ranges},
		\[
		\mathsf{Gen}_x^{\mathcal{M}}(s,C)
		\subseteq
		\mathcal{N}_x^{\mathcal{M}}(s,C),
		\qquad
		\mathsf{Gen}_x^{\mathcal{M}}(s,E)
		\subseteq
		\mathcal{N}_x^{\mathcal{M}}(s,E).
		\]
		Apply Theorem~\ref{thm:strategic-transversal-polarity}.
	\end{proof}
	
	\subsection{Boundary Transversal Duality}
	\label{subsec:boundary-duality}
	
	For arbitrary disjoint coalitions, transversal polarity gives an inclusion. For the boundary coalitions $\varnothing$ and $\Ag$, this inclusion becomes an exact duality.
	
	\begin{definition}[One-step value range]
		\label{def:one-step-value-range}
		Let $s\in S$ and $x\in X$. The \emph{one-step possible value range} of $x$ at $s$ is
		\[
		R_x^{\mathcal{M}}(s)
		\coloneqq
		\{
		\pi(o(s,\gamma))(x)
		\mid
		\gamma\in\Sigma_{\Ag}
		\}.
		\]
	\end{definition}
	
	\begin{theorem}[Boundary transversal duality]
		\label{thm:boundary-transversal-duality}
		Let $s\in S$ and $x\in X$. Then
		\[
		\mathcal{N}_x^{\mathcal{M}}(s,\varnothing)
		=
		\{\,A\subseteq D_x
		\mid
		R_x^{\mathcal{M}}(s)\subseteq A
		\,\},
		\]
		and
		\[
		\mathcal{N}_x^{\mathcal{M}}(s,\Ag)
		=
		\{\,A\subseteq D_x
		\mid
		A\cap R_x^{\mathcal{M}}(s)\neq\varnothing
		\,\}.
		\]
		Consequently,
		\[
		\mathcal{N}_x^{\mathcal{M}}(s,\Ag)
		=
		\bigl(\mathcal{N}_x^{\mathcal{M}}(s,\varnothing)\bigr)^{\#},
		\]
		and
		\[
		\mathcal{N}_x^{\mathcal{M}}(s,\varnothing)
		=
		\bigl(\mathcal{N}_x^{\mathcal{M}}(s,\Ag)\bigr)^{\#}.
		\]
	\end{theorem}
	
	\begin{proof}
		For the empty coalition, there is a unique empty action. It guarantees $A$ precisely when every complete action profile leads to a successor whose $x$-value lies in $A$. Hence
		\[
		A\in\mathcal{N}_x^{\mathcal{M}}(s,\varnothing)
		\]
		iff
		\[
		R_x^{\mathcal{M}}(s)\subseteq A.
		\]
		
		For the grand coalition, there is no opposing coalition. Thus
		\[
		A\in\mathcal{N}_x^{\mathcal{M}}(s,\Ag)
		\]
		iff there exists a complete profile $\gamma\in\Sigma_{\Ag}$ such that
		\[
		\pi(o(s,\gamma))(x)\in A,
		\]
		which is equivalent to
		\[
		A\cap R_x^{\mathcal{M}}(s)\neq\varnothing.
		\]
		
		It remains to verify the two transversal identities. Since all action sets are non-empty,
		\[
		R_x^{\mathcal{M}}(s)\neq\varnothing.
		\]
		The transversal dual of the principal upset
		\[
		\{\,A\subseteq D_x\mid R_x^{\mathcal{M}}(s)\subseteq A\,\}
		\]
		is exactly the family of subsets of $D_x$ intersecting $R_x^{\mathcal{M}}(s)$. Indeed, a set intersects every member of this principal upset iff it intersects its least member $R_x^{\mathcal{M}}(s)$.
		
		Conversely, the transversal dual of
		\[
		\{\,A\subseteq D_x\mid A\cap R_x^{\mathcal{M}}(s)\neq\varnothing\,\}
		\]
		is the family of subsets containing $R_x^{\mathcal{M}}(s)$. If $B$ intersects every subset meeting $R_x^{\mathcal{M}}(s)$, then in particular $B$ intersects each singleton $\{r\}$ with $r\in R_x^{\mathcal{M}}(s)$, so $R_x^{\mathcal{M}}(s)\subseteq B$. The converse is immediate.
	\end{proof}
	
	\subsection{Residual Value Indeterminacy}
	\label{subsec:residual-indeterminacy}
	
	Projected value-effectivity also measures how much value uncertainty remains after a coalition has chosen an optimal value-restricting action.
	
	\begin{definition}[Residual value indeterminacy]
		\label{def:residual-indeterminacy}
		Let $s\in S$, $x\in X$, and $C\subseteq\Ag$. The \emph{residual value indeterminacy} of $C$ over $x$ at $s$ is
		\[
		\iota_x^{\mathcal{M}}(s,C)
		\coloneqq
		\min\{\,|A|
		\mid
		A\in\mathcal{N}_x^{\mathcal{M}}(s,C)
		\,\}.
		\]
	\end{definition}
	
	The minimum is well-defined because $D_x$ is finite,
	\[
	D_x\in\mathcal{N}_x^{\mathcal{M}}(s,C),
	\]
	and
	\[
	\varnothing\notin\mathcal{N}_x^{\mathcal{M}}(s,C).
	\]
	
	\begin{theorem}[Operational characterisation of residual indeterminacy]
		\label{thm:operational-residual-indeterminacy}
		Let $s\in S$, $x\in X$, and $C\subseteq\Ag$. Then
		\[
		\iota_x^{\mathcal{M}}(s,C)
		=
		\min_{\alpha_C\in\Sigma_C}
		\left|
		R_x^C(s,\alpha_C)
		\right|.
		\]
		Equivalently,
		\[
		\iota_x^{\mathcal{M}}(s,C)
		=
		\min
		\{\,|R| \mid R\in\mathcal{R}_x^{\mathcal{M}}(s,C)\,\}.
		\]
	\end{theorem}
	
	\begin{proof}
		By Theorem~\ref{thm:operational-hypergraph-representation},
		\[
		\mathcal{N}_x^{\mathcal{M}}(s,C)
		=
		\uparrow\mathcal{R}_x^{\mathcal{M}}(s,C).
		\]
		Thus every enforceable region contains a strategic value range generated by some $C$-action, and every strategic value range is itself enforceable. Therefore the least cardinality of an enforceable value region is exactly the least cardinality of a strategic value range:
		\[
		\min\{\,|A| \mid A\in\mathcal{N}_x^{\mathcal{M}}(s,C)\,\}
		=
		\min\{\,|R| \mid R\in\mathcal{R}_x^{\mathcal{M}}(s,C)\,\}.
		\]
		Expanding the definition of $\mathcal{R}_x^{\mathcal{M}}(s,C)$ gives the displayed formula.
	\end{proof}
	
	\begin{proposition}[Basic properties of residual indeterminacy]
		\label{prop:residual-indeterminacy}
		Let $s\in S$ and $x\in X$. For every coalition $C\subseteq\Ag$,
		\[
		1\leq
		\iota_x^{\mathcal{M}}(s,C)
		\leq
		|D_x|.
		\]
		If $C\subseteq E$, then
		\[
		\iota_x^{\mathcal{M}}(s,E)
		\leq
		\iota_x^{\mathcal{M}}(s,C).
		\]
		Moreover,
		\[
		\iota_x^{\mathcal{M}}(s,C)=1
		\]
		if and only if there exists $c\in D_x$ such that
		\[
		\mathcal{M},s\models[C](x{=}c).
		\]
		In particular,
		\[
		\iota_x^{\mathcal{M}}(s,\varnothing)
		=
		|R_x^{\mathcal{M}}(s)|,
		\qquad
		\iota_x^{\mathcal{M}}(s,\Ag)=1.
		\]
	\end{proposition}
	
	\begin{proof}
		The bounds follow from Theorem~\ref{thm:projected-value-effectivity}\textup{(iii)}:
		\[
		D_x\in\mathcal{N}_x^{\mathcal{M}}(s,C)
		\quad\text{and}\quad
		\varnothing\notin\mathcal{N}_x^{\mathcal{M}}(s,C).
		\]
		If $C\subseteq E$, then
		\[
		\mathcal{N}_x^{\mathcal{M}}(s,C)
		\subseteq
		\mathcal{N}_x^{\mathcal{M}}(s,E)
		\]
		by Theorem~\ref{thm:projected-value-effectivity}\textup{(ii)}. Taking the minimum of cardinalities over a larger family cannot increase the value, so
		\[
		\iota_x^{\mathcal{M}}(s,E)
		\leq
		\iota_x^{\mathcal{M}}(s,C).
		\]
		
		The equality
		\[
		\iota_x^{\mathcal{M}}(s,C)=1
		\]
		holds iff some singleton $\{c\}$ belongs to $\mathcal{N}_x^{\mathcal{M}}(s,C)$, which is equivalent to
		\[
		\mathcal{M},s\models[C](x{=}c).
		\]
		
		For $C=\varnothing$, the unique empty action has strategic value range
		\[
		R_x^{\varnothing}(s,\langle\rangle)=R_x^{\mathcal{M}}(s).
		\]
		Therefore Theorem~\ref{thm:operational-residual-indeterminacy} gives
		\[
		\iota_x^{\mathcal{M}}(s,\varnothing)
		=
		|R_x^{\mathcal{M}}(s)|.
		\]
		For $C=\Ag$, each complete action profile yields a singleton value range, because there is no complementary coalition left to vary. Since $\Sigma_{\Ag}$ is non-empty,
		\[
		\iota_x^{\mathcal{M}}(s,\Ag)=1.
		\]
	\end{proof}
	
	Thus $\iota_x^{\mathcal{M}}(s,C)$ measures the remaining uncertainty about the next value of $x$ after coalition $C$ has chosen an optimal value-restricting action. The case $\iota_x^{\mathcal{M}}(s,C)=1$ is exact singleton value control; larger values correspond to genuinely set-valued control.
	
	\subsection{De Re and De Dicto Value Control}
	\label{subsec:dere-dedicto}
	
	The value structure also supports a simple analogue of the classical \emph{de re}/\emph{de dicto} distinction. The issue is whether a coalition can commit to a particular value in advance, or only guarantee that the realised value will satisfy a value-indexed condition.
	
	\begin{definition}[De re and de dicto value control]
		\label{def:dere-dedicto}
		Let $C\subseteq\Ag$, let $x\in X$, and let
		\[
		\{\psi_c\}_{c\in D_x}
		\]
		be a family of formulas. The \emph{de re} formula is
		\[
		\bigvee_{c\in D_x}
		[C]\bigl((x{=}c)\land\psi_c\bigr),
		\]
		whereas the \emph{de dicto} formula is
		\[
		[C]\bigvee_{c\in D_x}
		\bigl((x{=}c)\land\psi_c\bigr).
		\]
	\end{definition}
	
	The de re formula says that coalition $C$ can choose a particular value $c$ in advance and force both $x{=}c$ and $\psi_c$. The de dicto formula says only that $C$ can force the value-indexed disjunction; the realised value may still depend on the complementary coalition's response.
	
	\begin{theorem}[De re implies de dicto]
		\label{thm:dere-dedicto}
		For every $C\subseteq\Ag$, $x\in X$, and family $\{\psi_c\}_{c\in D_x}$,
		\[
		\models_{\VCL,\Theta}
		\left(
		\bigvee_{c\in D_x}[C]\bigl((x{=}c)\land\psi_c\bigr)
		\right)
		\to
		[C]\bigvee_{c\in D_x}\bigl((x{=}c)\land\psi_c\bigr).
		\]
		The converse is not valid in general.
	\end{theorem}
	
	\begin{proof}
		Suppose
		\[
		\mathcal{M},s\models
		[C]\bigl((x{=}c_0)\land\psi_{c_0}\bigr)
		\]
		for some $c_0\in D_x$. Since
		\[
		(x{=}c_0)\land\psi_{c_0}
		\]
		implies
		\[
		\bigvee_{c\in D_x}
		\bigl((x{=}c)\land\psi_c\bigr),
		\]
		the same witnessing $C$-action guarantees the latter disjunction. Hence the displayed implication is valid.
		
		For failure of the converse, take
		\[
		\Ag=\{1,2\},
		\qquad
		X=\{x\},
		\qquad
		D_x=\{a,b,c\},
		\]
		and let $C=\{1\}$. Let
		\[
		\Sigma_1=\{\alpha\},
		\qquad
		\Sigma_2=\{\beta_a,\beta_b\}.
		\]
		Consider states
		\[
		S=\{s,s_a,s_b\}.
		\]
		Define the outcome function at $s$ by
		\[
		o(s,\alpha\sqcup\beta_a)=s_a,
		\qquad
		o(s,\alpha\sqcup\beta_b)=s_b,
		\]
		and complete it by self-loops at $s_a$ and $s_b$. Let
		\[
		\pi(s_a)(x)=a,
		\qquad
		\pi(s_b)(x)=b,
		\qquad
		\pi(s)(x)=c.
		\]
		Finally set
		\[
		\psi_a=\top,
		\qquad
		\psi_b=\top,
		\qquad
		\psi_c=\bot.
		\]
		
		Then at $s$, coalition $\{1\}$ has its unique action $\alpha$, and every response of agent $2$ leads to a state satisfying
		\[
		((x{=}a)\land\psi_a)
		\lor
		((x{=}b)\land\psi_b)
		\lor
		((x{=}c)\land\psi_c).
		\]
		Therefore
		\[
		\mathcal{M},s\models
		[\,\{1\}\,]
		\bigvee_{d\in D_x}
		\bigl((x{=}d)\land\psi_d\bigr).
		\]
		However,
		\[
		\mathcal{M},s\not\models[\,\{1\}\,](x{=}a),
		\qquad
		\mathcal{M},s\not\models[\,\{1\}\,](x{=}b),
		\]
		because agent $2$ can choose the response leading to the other value. Also,
		\[
		\mathcal{M},s\not\models[\,\{1\}\,]\bigl((x{=}c)\land\psi_c\bigr),
		\]
		since $\psi_c=\bot$. Hence the de dicto formula holds at $s$, while the de re formula fails.
	\end{proof}
	
	The failure of the converse is the familiar failure of finite additivity for coalition ability:
	\[
	[C](x\in A)
	\not\Rightarrow
	\bigvee_{c\in A}[C](x{=}c).
	\]
	A coalition may force the next value of $x$ into a finite region $A$ without being able to force any particular value in $A$. In terms of residual indeterminacy, this is precisely the difference between set-valued control and singleton value control.

	\section{Conclusion and Future Work}
	\label{sec:conclusion}
	
	This paper introduced Value Coalition Logic ($\VCL$), a typed assignment-based reconstruction of classical Coalition Logic. The language replaces flat propositional atoms with sorted value atoms $(x{=}c)$ over finite domains, while preserving the ordinary one-step semantics of coalitional ability. Thus the strategic meaning of $[C]\varphi$ is unchanged: coalition $C$ has a joint action guaranteeing $\varphi$ against all actions of the complementary coalition. The change lies in the atomic state-description layer. Each state carries a total typed assignment, so exhaustivity and mutual exclusion of alternative values are built into the semantics rather than imposed as external propositional coherence constraints.
	
	We proved that, over a fixed finite typed signature, $\VCL$ corresponds exactly to propositional Coalition Logic over coherent valuations. Under the translation
	\[
	(x{=}c)\mapsto p_x^c,
	\]
	truth is preserved and reflected between $\VCL$ models and coherent propositional coalition models over the same strategic component. The two model constructions are inverse on the intended coherent class. This correspondence yields a Hilbert-style axiomatisation obtained by adding finite-domain value-coherence axioms to the standard axioms of Coalition Logic. In this precise sense, $\VCL$ is conservative at the level of the coalition modality: it does not add a new strategic semantics, but isolates the coherent finite-value fragment of the propositional presentation and treats it as typed semantic structure.
	
	The main structural contribution was the value-level analysis of ordinary coalitional ability. For each state $s$ and variable $x$, the underlying one-step game form induces an $x$-value quotient, or value-projection, game form over $D_x$. Projected value-effectivity families are exactly the ordinary effectivity families of these quotient game forms. Equivalently, they are the upward closures of the strategic value ranges generated by concrete coalition actions. This gives an operational hypergraph representation of value control: a coalition can enforce a value region precisely when it has an action whose possible next $x$-values are all contained in that region.
	
	From this representation we derived set-valued strategic exclusion, transversal polarity between disjoint coalitions, exact boundary duality between $\varnothing$ and $\Ag$, residual value indeterminacy, and a de re/de dicto contrast for value control. These results are not additional modal principles. They are value-level invariants inherited from ordinary one-step game-form effectivity after projection to finite value domains.
	
	Several directions remain open. First, one may extend the language with controlled forms of quantification over finite value domains. Such an extension would allow the de re/de dicto distinction to be expressed more directly in the object language, separating the ability to enforce a condition for some fixed value from the ability to ensure that some suitable value is realised.
	
	Second, the value-quotient construction suggests a representation problem for projected value-effectivity. One may ask which systems of upward-closed families over $\mathcal{P}(D_x)$ can arise from a single finite value-quotient game form, and how their minimal generators, transversal duals, and residual-indeterminacy indices classify different patterns of value control.
	
	Third, the present paper analysed one variable at a time. A natural next step is to study multi-variable value-effectivity, where coalitions enforce partial assignments or regions over products of value domains. This would connect the framework with dependency analysis, constraint-based specifications, and structured models of multi-agent systems. Typed value atoms could also be combined with temporal, epistemic, imperfect-information, quantitative, or resource-sensitive extensions of Coalition Logic, yielding systems for reasoning about value trajectories, partial observability, costs, preferences, and long-term strategic control.


\begin{thebibliography}{99}
	
	\bibitem{Pauly02}
	Marc Pauly.
	\newblock A modal logic for coalitional power in games.
	\newblock \emph{Journal of Logic and Computation}, 12(1):149--166, 2002.
	\newblock doi:\href{https://doi.org/10.1093/logcom/12.1.149}{10.1093/logcom/12.1.149}.
	
	\bibitem{Pauly02Comp}
	Marc Pauly.
	\newblock On the complexity of coalitional reasoning.
	\newblock \emph{International Game Theory Review}, 4(3):237--254, 2002.
	\newblock doi:\href{https://doi.org/10.1142/S0219198902000677}{10.1142/S0219198902000677}.
	
	\bibitem{ATL02}
	Rajeev Alur, Thomas A.\ Henzinger, and Orna Kupferman.
	\newblock Alternating-time temporal logic.
	\newblock \emph{Journal of the ACM}, 49(5):672--713, 2002.
	\newblock doi:\href{https://doi.org/10.1145/585265.585270}{10.1145/585265.585270}.
	
	\bibitem{CHP10}
	Krishnendu Chatterjee, Thomas A.\ Henzinger, and Nir Piterman.
	\newblock Strategy logic.
	\newblock \emph{Information and Computation}, 208(6):677--693, 2010.
	\newblock doi:\href{https://doi.org/10.1016/j.ic.2009.07.004}{10.1016/j.ic.2009.07.004}.
	
	\bibitem{MMPV14}
	Fabio Mogavero, Aniello Murano, Giuseppe Perelli, and Moshe~Y.\ Vardi.
	\newblock Reasoning about strategies: On the model-checking problem.
	\newblock \emph{ACM Transactions on Computational Logic}, 15(4):34:1--34:47, 2014.
	\newblock doi:\href{https://doi.org/10.1145/2631917}{10.1145/2631917}.
	
	\bibitem{HHMW01}
	Harrenstein, P., van der Hoek, W., Meyer, J. J., and Witteveen, C.
	Boolean games.
	In \emph{Proceedings of the 8th Conference on Theoretical Aspects of Rationality and Knowledge}, pages 287--298, 2001.
	
	\bibitem{BLLZ06}
	Elise Bonzon, Marie-Christine Lagasquie-Schiex, Jérôme Lang, and Bruno Zanuttini.
	\newblock Boolean games revisited.
	\newblock In \emph{ECAI 2006: 17th European Conference on Artificial Intelligence}, volume 141, page 265. SAGE Publications Limited, 2006.
	
	\bibitem{DHKW08}
	Paul~E.\ Dunne, Wiebe van der Hoek, Sarit Kraus, and Michael Wooldridge.
	\newblock Cooperative boolean games.
	\newblock In \emph{Proceedings of the 7th international joint conference on Autonomous agents and multiagent systems-Volume 2}, pages 1015--1022. 2008.
	
	\bibitem{CLPC05}
	Wiebe van der Hoek and Michael~J.\ Wooldridge.
	\newblock On the logic of cooperation and propositional control.
	\newblock \emph{Artificial Intelligence}, 164(1--2):81--119, 2005.
	\newblock doi:\href{https://doi.org/10.1016/j.artint.2005.01.003}{10.1016/j.artint.2005.01.003}.
	
	\bibitem{Gerbrandy06}
	Jelle Gerbrandy.
	\newblock Logics of propositional control.
	\newblock In \emph{Proceedings of the 5th International Joint Conference on Autonomous Agents and Multiagent Systems (AAMAS 2006)}, pages 193--200. ACM, 2006.
	
	\bibitem{BH16}
	Francesco Belardinelli and Andreas Herzig.
	\newblock On logics of strategic ability based on propositional control.
	\newblock In \emph{Proceedings of the 25th International Joint Conference on Artificial Intelligence (IJCAI 2016)}, pages 95--101. IJCAI/AAAI Press, 2016.
	
	\bibitem{GJT13}
	Valentin Goranko, Wojciech Jamroga, and Paolo Turrini.
	\newblock Strategic games and truly playable effectivity functions.
	\newblock \emph{Autonomous Agents and Multi-Agent Systems}, 26(2):288--314, 2013.
	\newblock doi:\href{https://doi.org/10.1007/s10458-012-9192-y}{10.1007/s10458-012-9192-y}.
	
	\bibitem{MMPV17}
	Fabio Mogavero, Aniello Murano, Giuseppe Perelli, and Moshe~Y.\ Vardi.
	\newblock Reasoning about strategies: On the satisfiability problem.
	\newblock \emph{Logical Methods in Computer Science}, 13(1):1--37, 2017.
	\newblock doi:\href{https://doi.org/10.23638/LMCS-13(1:9)2017}{10.23638/LMCS-13(1:9)2017}.
	
	\bibitem{AGJ07}
	Thomas {\AA}gotnes, Valentin Goranko, and Wojciech Jamroga.
	\newblock Alternating-time temporal logics with irrevocable strategies.
	\newblock In \emph{Proceedings of the 11th Conference on Theoretical Aspects of Rationality and Knowledge (TARK 2007)}, pages 15--24, 2007.
	
	\bibitem{BJ14}
	Nils Bulling and Wojciech Jamroga.
	\newblock Comparing variants of strategic ability: How uncertainty and memory influence general properties of games.
	\newblock \emph{Autonomous Agents and Multi-Agent Systems}, 28(3):474--518, 2014.
	\newblock doi:\href{https://doi.org/10.1007/s10458-013-9231-3}{10.1007/s10458-013-9231-3}.
	
	\bibitem{BMMRV17}
	Rapha{\"e}l Berthon, Bastien Maubert, Aniello Murano, Sasha Rubin, and Moshe~Y.\ Vardi.
	\newblock Strategy logic with imperfect information.
	\newblock In \emph{Proceedings of the 32nd Annual ACM/IEEE Symposium on Logic in Computer Science (LICS 2017)}, pages 1--12. IEEE Computer Society, 2017.
	\newblock doi:\href{https://doi.org/10.1109/LICS.2017.8005136}{10.1109/LICS.2017.8005136}.
	
	\bibitem{CLMM18}
	Petr Cerm{\'a}k, Alessio Lomuscio, Fabio Mogavero, and Aniello Murano.
	\newblock Practical verification of multi-agent systems against {SLK} specifications.
	\newblock \emph{Information and Computation}, 261:588--614, 2018.
	\newblock doi:\href{https://doi.org/10.1016/j.ic.2017.09.011}{10.1016/j.ic.2017.09.011}.
	
	\bibitem{BLLMY22}
	Francesco Belardinelli, Alessio Lomuscio, Vadim Malvone, and Emily Yu.
	\newblock Approximating perfect recall when model checking strategic abilities: Theory and applications.
	\newblock \emph{Journal of Artificial Intelligence Research}, 73:897--932, 2022.
	\newblock doi:\href{https://doi.org/10.1613/jair.1.12539}{10.1613/jair.1.12539}.
	
	\bibitem{KJ24}
	Damian Kurpiewski, Mateusz Kami{\'n}ski, and Wojciech Jamroga.
	\newblock {STV+FLY}: On-the-fly model checking of strategic ability in multi-agent systems.
	\newblock In \emph{ECAI 2024 --- 27th European Conference on Artificial Intelligence}, pages 4483--4486. IOS Press, 2024.
	\newblock doi:\href{https://doi.org/10.3233/FAIA241035}{10.3233/FAIA241035}.
	
	\bibitem{BFJMM23}
	Francesco Belardinelli, Angelo Ferrando, Wojciech Jamroga, Vadim Malvone, and Aniello Murano.
	\newblock Scalable verification of strategy logic through three-valued abstraction.
	\newblock In \emph{Proceedings of the 32nd International Joint Conference on Artificial Intelligence (IJCAI 2023)}, pages 46--54. IJCAI, 2023.
	\newblock doi:\href{https://doi.org/10.24963/ijcai.2023/6}{10.24963/ijcai.2023/6}.
	
	\bibitem{BJMM19}
	Francesco Belardinelli, Wojciech Jamroga, Vadim Malvone, and Aniello Murano.
	\newblock Strategy logic with simple goals: Tractable reasoning about strategies.
	\newblock In \emph{28th International Joint Conference on Artificial Intelligence (IJCAI 2019)}, pages 88--94. 2019.
	
	\bibitem{JKKP20}
	Wojciech Jamroga, Beata Konikowska, Damian Kurpiewski, and Wojciech Penczek.
	\newblock Multi-valued verification of strategic ability.
	\newblock \emph{Fundamenta Informaticae}, 175(1--4):207--251, 2020.
	\newblock doi:\href{https://doi.org/10.3233/FI-2020-1955}{10.3233/FI-2020-1955}.
	
	\bibitem{BG22}
	Nils Bulling and Valentin Goranko.
	\newblock Combining quantitative and qualitative reasoning in concurrent multi-player games.
	\newblock \emph{Autonomous Agents and Multi-Agent Systems}, 36:Article~2, 2022.
	\newblock doi:\href{https://doi.org/10.1007/s10458-021-09531-9}{10.1007/s10458-021-09531-9}.
	
		
	\bibitem{LXL20}
	Zhaoshuai Liu, Liping Xiong, Yongmei Liu, Yves Lespérance, Ronghai Xu, and Hongyi Shi.
	\newblock A Modal Logic for Joint Abilities under Strategy Commitments.
	\newblock In \emph{Proceedings of the Twenty-Ninth International Joint Conference on Artificial Intelligence (IJCAI 2020)}, pages 1805--1812. 2020.
	\newblock doi:\href{https://doi.org/10.24963/ijcai.2020/250}{10.24963/ijcai.2020/250}.
	
	\bibitem{Alechina09RBCL}
	Natasha Alechina, Brian Logan, Hoang Nga Nguyen, and Abdur Rakib.
	\newblock Logic for coalitions with bounded resources.
	\newblock \emph{Journal of Logic and Computation}, 21(6):907--937, 2011.
	
	\bibitem{Alechina10RBATL}
	Natasha Alechina, Brian Logan, Nguyen Hoang Nga, and Abdur Rakib.
	\newblock Resource-bounded alternating-time temporal logic.
	\newblock In \emph{Proceedings of the 9th International Conference on Autonomous Agents and Multiagent Systems: volume 1-Volume 1}, pages 481--488. 2010.
	
\end{thebibliography}
\end{document}